\documentclass[11pt,a4paper]{article}
\usepackage{jheppub}

\usepackage{amsmath,amssymb,amsthm,mathtools}
\usepackage{bm}
\usepackage{tikz-cd}
\usepackage{booktabs}
\usepackage{array}

\hypersetup{colorlinks=true,linkcolor=blue,citecolor=blue,urlcolor=blue}

\newtheorem{theorem}{Theorem}[section]
\newtheorem{proposition}[theorem]{Proposition}

\newtheorem{corollary}[theorem]{Corollary}
\newtheorem{definition}[theorem]{Definition}
\newtheorem{assumption}[theorem]{Assumption}
\newtheorem{remark}[theorem]{Remark}
\newtheorem{construction}[theorem]{Construction}
\newtheorem{example}[theorem]{Example}

\makeatletter
\@ifundefined{widetext}{\newenvironment{widetext}{\par\medskip}{\par\medskip}}{}
\makeatother

\newcommand{\V}{\mathcal V}
\newcommand{\Van}{\mathcal V_{\mathrm{an}}}
\newcommand{\Valg}{\mathcal V_{\mathrm{alg}}}
\newcommand{\Fil}{\mathrm{Fil}}
\newcommand{\Fact}{\mathrm{Fact}}

\newcommand{\Disk}{\mathrm{Disk}}

\newcommand{\Core}{\mathrm{Core}}
\newcommand{\DIS}{\mathrm{DIS}}
\newcommand{\coll}{\mathrm{coll}}
\newcommand{\hard}{\mathrm{hard}}

\newcommand{\Meas}{\mathrm{Meas}}
\newcommand{\spc}{\mathrm{sp}_{\coll}}
\newcommand{\id}{\mathrm{id}}
\newcommand{\coker}{\operatorname{coker}}

\newcommand{\im}{\operatorname{im}}

\newcommand{\Hom}{\operatorname{Hom}}

\newcommand{\Q}{\mathbb Q}
\newcommand{\R}{\mathbb R}

\newcommand{\hot}{\widehat\otimes}

\newcommand{\bO}{\mathcal O}
\newcommand{\A}{\mathcal A}
\newcommand{\Rreg}{\mathcal R}
\newcommand{\Ucoll}{U_{\coll}}

\newcommand{\order}{\mathcal O}
\newcommand{\pp}{\mathrm{pp}}
\newcommand{\Coeff}{\mathsf C}
\newcommand{\PDF}{\mathsf F}
\newcommand{\Sch}{\mathsf S}
\newcommand{\PSS}{\mathsf{PSS}}
\newcommand{\Lan}{\mathsf L}

\newcommand{\MS}{\overline{\mathrm{MS}}}

\DeclareMathOperator*{\coeq}{coeq}

\DeclareMathOperator{\Star}{Star}
\DeclareMathOperator{\Link}{Link}
\DeclareMathOperator{\Cone}{Cone}

\title{Scheme-invariant stratified factorization algebras for inclusive deep inelastic scattering}

\author{Dustin Keller}
\affiliation{University of Virginia}

\emailAdd{dustin@virginia.edu}

\abstract{Inclusive deep inelastic scattering factorization combines two features that are often treated separately: an asymptotic reconstruction of the current-current matrix element from hard and long-distance data, and an invariance under finite changes of collinear scheme or operator basis.  We formulate these two features as a single proof object.  The construction packages the leading-region analysis, overlap subtraction, Wilson-line reduction, finite scheme kernels and physical measurement into a typed, filtered structure on a compactified space of asymptotic regimes.  Its central carrier is the balanced hard-collinear core over the interface algebra of finite scheme transformations.  The hard QCD input is the construction of a scheme-balanced comparison map from this core to the collinear collar of the regime algebra.  Once this comparison is an equivalence through the chosen power accuracy and the measurement descends to convolution, the standard DIS convolution formula follows formally and independently of the chosen scheme presentation.  We separate this formal implication from the analytic QCD obligations needed to construct the collar equivalence, describe Collins-style subtraction as descent and M\"obius inversion on the region poset, and give a finite check relating $\overline{\mathrm{MS}}$ and DIS presentations.  The framework is intended as proof infrastructure rather than as a new calculation of DIS coefficient functions.  It supplies diagnostics for missing regions, nonclosed operator sectors, nonbalanced measurements and failed collar equivalences, and it gives a typed interface for future proof-assistant and machine-learning implementations of factorization workflows.}


\begin{document}

\maketitle

\section{Introduction}
\label{sec:introduction}

Inclusive deep inelastic scattering (DIS) has long served as the cleanest arena in which short-distance QCD, parton distributions, and operator product methods meet.  In the Bjorken limit,
\begin{equation}
Q^2=-q^2\to \infty,\qquad x_B=\frac{Q^2}{2P\cdot q}\quad \text{fixed},
\end{equation}
structure functions admit an asymptotic expansion whose leading-power term has the familiar collinear-factorized form
\begin{align}
F_i(x_B,Q^2)
&=\sum_{a\in \mathcal I}\int_{x_B}^1\frac{dz}{z}\,
\nonumber\\
&\quad\times C_{i,a}(z,Q,\mu)
\nonumber\\
&\quad\times f_{a/H}\!\left(\frac{x_B}{z},\mu\right)
+R_i(x_B,Q^2),
\label{eq:intro-standard}
\end{align}
with $R_i=\order((\Lambda_{\mathrm{QCD}}/Q)^{p})$ at the claimed power accuracy.  The coefficient functions $C_{i,a}$ are short-distance quantities, while $f_{a/H}$ are universal long-distance parton distribution functions (PDFs).  Equation~\eqref{eq:intro-standard} is the practical engine behind much of QCD phenomenology, but as a mathematical statement it suppresses several pieces of structure that are essential in an all-orders proof.

First, by a DIS factorization proof we mean the all-orders asymptotic construction that begins with the forward Compton amplitude, identifies the pinch and scaling regions of renormalized graphs, constructs approximators in those regions, subtracts their overlaps, uses Ward identities to reduce collinear attachments to Wilson lines, and bounds the remainder.  It is not simply a manipulation of the final convolution formula.  In this construction the graph expansion is organized by hard, collinear, soft, jet and overlap configurations.  This region-theoretic view appears already in the Libby-Sterman analysis of high-energy graphs and in the Collins-Soper-Sterman factorization program, and is developed systematically in modern treatments of perturbative QCD factorization~\cite{LibbySterman1978a,LibbySterman1978b,CSS1989,CSS2004,Collins2011,Sterman1993}.  Related asymptotic-expansion methods, including sectors, regions, and blowup-like decompositions of singular momentum configurations, also appear in the literature on Feynman integrals and the method of regions~\cite{Hepp1966,Speer1975,BenekeSmirnov1998,Smirnov2002}.  In DIS, Ward identities and eikonal approximations turn suitable classes of gluon attachments into Wilson lines, producing gauge-invariant PDF operators~\cite{CollinsSoper1982,Collins2011}.  Inclusivity then plays a decisive role in the cancellation or absorption of soft final-state effects.

Second, the two factors in Eq.~\eqref{eq:intro-standard} are not separately physical.  Finite changes of factorization scheme or finite renormalized-operator basis changes transform the PDF multiplet and coefficient multiplet inversely,
\begin{equation}
\bm f' = Z\circledast \bm f,
\qquad
\bm C'_i=\bm C_i\circledast Z^{-1},
\label{eq:intro-scheme-change}
\end{equation}
leaving $\bm C_i\circledast \bm f$ invariant through the stated accuracy.  The scheme dependence is present already in the leading-twist operator product expansion (OPE) and in the perturbative evolution of parton densities, whose early development includes the work of Wilson, Gross and Wilczek, Georgi and Politzer, Gribov and Lipatov, Altarelli and Parisi, Dokshitzer, and Curci-Furmanski-Petronzio~\cite{Wilson1969,GrossWilczek1974,GeorgiPolitzer1976,GribovLipatov1972,AltarelliParisi1977,Dokshitzer1977,CurciFurmanskiPetronzio1980}.  Thus a formal statement of factorization should distinguish the presentation-dependent pair $(\bm C,\bm f)$ from the invariant composite that it presents.

Third, a factorization theorem is reusable only if its proof data are separated from its presentation.  Standard approaches such as CSS factorization, OPE matching, and soft-collinear effective theory (SCET) give powerful process-specific languages for constructing coefficients, operators, and evolution kernels~\cite{BauerFlemingLuke2000,BauerFlemingPirjolStewart2001,BauerPirjolStewart2002}.  The aim here is not to replace these analytic or effective-theory tools.  Rather, the aim is to identify the typed mathematical object that their successful use constructs.  Once that object exists, factorization, scheme invariance, and measurement compatibility become formal consequences of its maps and universal properties.  This typed viewpoint is also useful when some of the objects are numerical or learned rather than written in closed analytic form.  Modern PDF extractions already use neural networks and deep-learning-inspired optimization as flexible parametrizations of long-distance functions~\cite{ForteGarridoLatorrePiccione2002,DelDebbioEtAl2007NNPDF,BallEtAl2009NNPDF,CarrazzaCruzMartinez2019,BallEtAl2022NNPDF40}; the present framework specifies the module actions, balancing relations and measurement maps that such parametrizations must respect.  In this sense it is compatible with categorical approaches to deep learning, where architectures are treated as typed compositional maps subject to algebraic constraints~\cite{GavranovicEtAl2024CategoricalDL}.

The central object proposed in this paper is a \emph{DIS factorization proof object}.  At order $N$ it consists of
\begin{equation}
(K_{\DIS},\Rreg_{\DIS}^{\pp},\A_{\DIS},\Sch^{\leq N},\Coeff^{\leq N},\PDF^{\leq N},
\Phi_{\coll}^{\Core},\Meas),
\label{eq:proof-object-intro}
\end{equation}
where $K_{\DIS}$ is a compactified stratified regime space, $\Rreg_{\DIS}^{\pp}$ is a proof-level region poset, $\A_{\DIS}\in\Fact(K_{\DIS};\Fil(\V))$ is a constructible filtered regime factorization algebra, $\Sch^{\leq N}$ is a collinear interface algebra, $\Coeff^{\leq N}$ and $\PDF^{\leq N}$ are right and left modules over it, $\Phi_{\coll}^{\Core}$ is a scheme-balanced collar comparison map, and $\Meas$ is the physical measurement to a convolution algebra.  We prove that a valid object of this kind implies the measured DIS factorization formula.  The analytic content of QCD is then concentrated in the verification that the particular proof object extracted from renormalized perturbative DIS is valid.

The proof object is not a new approximation to DIS.  It is a typed container for the data produced by a standard factorization proof.  A conventional proof supplies region types, approximators, subtraction maps, Ward-identity reductions, finite scheme kernels, hard coefficient data, PDF or light-ray correlator data, and a measurement projection.  The proof object asks whether these data close under finite scheme transformations, satisfy descent over overlaps, and define a collar equivalence through the claimed power accuracy.  If they do, the measured convolution formula follows formally.

On the geometric side, the finite pattern of leading regions and overlaps is represented by a compactified stratified regime space $K_{\DIS}$.  The hard regime is an open stratum, the incoming-collinear regime is a boundary stratum, and soft/jet/overlap regions appear as deeper strata in a proof-level refinement.  A constructible filtered factorization algebra
\begin{equation}
\A_{\DIS}\in \Fact(K_{\DIS};\Fil(\V))
\label{eq:intro-regime-algebra}
\end{equation}
assigns algebraic data to regime neighborhoods and records power accuracy through its filtration.  The notation is inspired by factorization algebras in quantum field theory and their local-to-global descent properties~\cite{CostelloGwilliam1,CostelloGwilliam2}.  Because the base space here is a stratified space of asymptotic regimes rather than physical spacetime, constructibility, conical neighborhoods, exit-path categories, and factorization-homology ideas provide useful language~\cite{KashiwaraSchapira1990,GoreskyMacPherson1988,AyalaFrancisTanaka2015,BeilinsonDrinfeld2004}.

On the scheme-invariance side, the interface-algebra concept of Ref.~\cite{Keller2026Core} is used.  At a fixed accuracy $N$, admissible finite collinear counterterms and operator-mixing kernels form an algebra object $\mathsf S^{\leq N}$.  Coefficients form a right module $\mathsf C^{\leq N}$, PDFs/correlators form a left module $\mathsf F^{\leq N}$, and the invariant carrier is the relative tensor product
\begin{equation}
\Core^{\leq N}_{\DIS}:=
\mathsf C^{\leq N}\otimes_{\mathsf S^{\leq N}}\mathsf F^{\leq N}.
\label{eq:intro-core}
\end{equation}
This quotient imposes the balancing relation
\begin{equation}
(c\cdot s)\otimes f\sim c\otimes (s\cdot f),
\label{eq:intro-balance}
\end{equation}
which is the categorical form of the inverse transformation in Eq.~\eqref{eq:intro-scheme-change}.  The resulting object is independent of the chosen presentation of the factorized expression.

The central map in this paper is the scheme-balanced collar comparison
\begin{equation}
\Phi^{\Core}_{\coll}:
\Core^{\leq N}_{\DIS}\longrightarrow \A_{\DIS}(\Ucoll),
\label{eq:intro-collar-map}
\end{equation}
where $\Ucoll$ is a collar neighborhood of the incoming-collinear boundary stratum.  We define order-$N$ DIS factorization to mean that Eq.~\eqref{eq:intro-collar-map} is an $N$-equivalence: after quotienting by the $(N+1)$st filtration piece, it becomes an isomorphism.  The hard analytic QCD work is precisely the construction of this map from leading-region approximators, Wilson-line reductions, soft cancellations, and overlap subtractions, and the proof that it is an $N$-equivalence.  Once that is available, the factorized DIS formula follows by a formal diagram chase.

The contribution of this work is a formal interface for building and validating factorization proofs.  Theorems below prove that a scheme-invariant collar equivalence plus a convolution-compatible measurement implies the usual DIS formula.  The descent section shows how Collins-style subtraction is naturally expressed as M\"obius inversion on the region poset.  The QCD realization section isolates the analytic assumptions whose verification would discharge the collar-equivalence obligation.  Sec.~\ref{sec:proof-protocol} makes the operational claim explicit: the framework is a proof protocol that separates QCD input from formal consistency checks.  Sec.~\ref{sec:scheme-demo} previews the formalism in a familiar finite setting without moving the example out of its natural place in the proof protocol: the usual change between an $\overline{\mathrm{MS}}$ presentation and a DIS-scheme presentation maps to the same element of the balanced core, while an insufficiently closed parton sector fails to define the core.  Sec.~\ref{subsec:obstruction-diagnostics} explains how failed proof objects can localize missing physics.  This organization is also compatible with computer-assisted formalization: the verification target is a finite list of objects, maps, filtrations, coequalizers, and commutative diagrams.

For orientation, Table~\ref{tab:dictionary} gives a physics-to-formalism dictionary.  The table is not an analogy; the entries on the right are the objects whose existence or compatibility is required by the definitions below.  The later proof-protocol table refines this dictionary into a sequence of construction and verification steps that a DIS factorization proof must complete.

\begin{table*}[t]
\caption{Dictionary between standard DIS factorization language and the proof-object language used in this paper.}
\label{tab:dictionary}
\begin{tabular}{@{}ll@{}}
\toprule
\parbox[t]{0.30\textwidth}{QCD/factorization language} & \parbox[t]{0.62\textwidth}{Formal object or condition}\\
\midrule
\parbox[t]{0.30\textwidth}{Hard coefficient functions $C_{i,a}$} & \parbox[t]{0.62\textwidth}{Right module $\Coeff^{\leq N}$ over the interface algebra $\Sch^{\leq N}$}\\[0.6ex]
\parbox[t]{0.30\textwidth}{PDFs and light-ray correlators $f_{a/H}$} & \parbox[t]{0.62\textwidth}{Left module $\PDF^{\leq N}$; closure under $\Sch^{\leq N}$ is part of the proof datum}\\[0.6ex]
\parbox[t]{0.30\textwidth}{Finite factorization-scheme change $Z$} & \parbox[t]{0.62\textwidth}{Invertible element/morphism of $\Sch^{\leq N}$ acting oppositely on $\Coeff^{\leq N}$ and $\PDF^{\leq N}$}\\[0.6ex]
\parbox[t]{0.30\textwidth}{Scheme-independent recomposition $\sum_a C_a\otimes f_a$} & \parbox[t]{0.62\textwidth}{Balanced core $\Coeff^{\leq N}\otimes_{\Sch^{\leq N}}\PDF^{\leq N}$}\\[0.6ex]
\parbox[t]{0.30\textwidth}{Leading regions and pinch surfaces} & \parbox[t]{0.62\textwidth}{Strata of $K_{\DIS}$, indexed by the proof-level poset $\Rreg^{\pp}_{\DIS}$}\\[0.6ex]
\parbox[t]{0.30\textwidth}{Overlap subtraction} & \parbox[t]{0.62\textwidth}{Descent/Cech compatibility and M\"obius inversion in the incidence algebra of $\Rreg^{\pp}_{\DIS}$}\\[0.6ex]
\parbox[t]{0.30\textwidth}{Power-suppressed remainder} & \parbox[t]{0.62\textwidth}{Filtration piece $F^{N+1}$ and quotient functor $\pi_{\leq N}$}\\[0.6ex]
\parbox[t]{0.30\textwidth}{Factorization theorem} & \parbox[t]{0.62\textwidth}{The balanced collar map $\Phi^{\Core}_{\coll}$ is an $N$-equivalence and measurement descends to the core}\\[0.6ex]
\parbox[t]{0.30\textwidth}{Incomplete or inconsistent proof} & \parbox[t]{0.62\textwidth}{Failure of a typed condition: missing stratum, nonclosed module sector, nonbalanced measurement, noncommuting descent diagram, or failed collar equivalence}\\
\bottomrule
\end{tabular}
\end{table*}

Readers primarily interested in QCD factorization can read Secs.~\ref{sec:dis}, \ref{sec:regime-construction}, \ref{sec:diagrammatic-realization}, \ref{sec:proof-protocol}, \ref{sec:scheme-demo}, \ref{sec:qcd-realization} and \ref{sec:discussion} first.  Readers primarily interested in the formal statement can read Secs.~\ref{sec:filtered}--\ref{sec:formal-theorem} and then return to the QCD realization in Sec.~\ref{sec:qcd-realization}.

\section{Inclusive DIS and scheme covariance}
\label{sec:dis}

\subsection{Kinematics and hadronic tensor}

Inclusive DIS is
\begin{equation}
\ell(l)+H(P)\to \ell(l')+X,
\end{equation}
with
\begin{equation}
q=l-l',\qquad Q^2=-q^2>0,
\qquad x_B=\frac{Q^2}{2P\cdot q}.
\end{equation}
The hadronic tensor is
\begin{equation}
W^{\mu\nu}(P,q)=\frac{1}{4\pi}\int d^4x\,e^{iq\cdot x}
\langle H(P)|[J^\mu(x),J^\nu(0)]|H(P)\rangle,
\label{eq:hadronic-tensor}
\end{equation}
or equivalently the discontinuity of the forward Compton amplitude.  Projection onto Lorentz structures gives scalar structure functions $F_i$.  For $x_B\in(0,1)$ away from endpoints and $Q\gg \Lambda_{\mathrm{QCD}}$, leading-twist collinear factorization gives Eq.~\eqref{eq:intro-standard} with parton channels $a\in\mathcal I$.

A quark PDF can be represented by the gauge-invariant light-ray matrix element
\begin{align}
 f_{q/H}(x,\mu)
&=\int\frac{d\lambda}{2\pi}\,e^{-ix\lambda P\cdot n}
\nonumber\\
&\quad\times\big\langle H(P)\big|
\bar\psi(\lambda n)W(\lambda n,0)
\nonumber\\
&\qquad \gamma\cdot n\,\psi(0)\big|H(P)\big\rangle_\mu .
\label{eq:pdf-def}
\end{align}
Here $n$ is lightlike and $W(\lambda n,0)$ is the Wilson line.  Gluon and polarized distributions are analogous.  The exact spin projection, normalization, and Wilson-line orientation are conventional; the formal structure below requires only that the retained correlators form a sector closed under the relevant finite collinear renormalizations.

\subsection{Parton multiplets and finite kernels}

Let $D$ denote a convolution algebra of distributions on the momentum-fraction monoid $G=(0,1]$, with multiplication
\begin{equation}
(u*v)(x)=\int_x^1\frac{dz}{z}\,u(z)v(x/z).
\label{eq:conv}
\end{equation}
For a finite channel set $\mathcal I$, write
\begin{equation}
\bm f=(f_a)_{a\in\mathcal I}\in D^{\mathcal I},
\qquad
\bm C_i=(C_{i,a})_{a\in\mathcal I}\in (D^{\mathcal I})^\vee.
\end{equation}
A finite scheme transformation is a matrix $Z=(Z_{ab})$ of kernels.  The left action on PDFs and the right action on coefficients are
\begin{equation}
\begin{aligned}
(Z\circledast \bm f)_a&=\sum_b Z_{ab}*f_b,\\
(\bm C_i\circledast Z)_b&=\sum_a C_{i,a}*Z_{ab}.
\end{aligned}
\label{eq:module-actions}
\end{equation}
If $Z$ is convolution-invertible in the retained sector, Eq.~\eqref{eq:intro-scheme-change} leaves
\begin{equation}
\langle \bm C_i,\bm f\rangle:=\sum_a C_{i,a}*f_a
\label{eq:pairing}
\end{equation}
invariant.  The relative tensor product used below is the universal object through which all such balanced pairings factor.

\section{Target categories, filtrations, and completed tensors}
\label{sec:filtered}

\subsection{Formal target}

Let $(\V,\otimes,\mathbf 1)$ be an additive symmetric monoidal category.  For the algebraic statements we assume that $\V$ has the finite colimits used below and that tensoring preserves the relevant cokernels and coequalizers separately in each variable.  When $\pi_{\leq N}$ is used as an exact truncation functor we assume the corresponding cokernels exist.  A model example for mechanized formalization is
\begin{equation}
\Valg=\mathrm{Ch}_k,
\end{equation}
chain complexes over a field $k$.  In this algebraic model all tensor products are ordinary algebraic tensor products.  We introduce this model now because the internal result of Sec.~\ref{sec:formal-theorem} uses only this algebraic layer: it is a theorem about maps, truncations and coequalizers, while the QCD sections supply the physical objects that populate those maps.

\begin{definition}[Filtered object]
A filtered object is an object $X\in\V$ with a decreasing sequence of subobjects
\begin{equation}
X=F^0X\supset F^1X\supset F^2X\supset\cdots .
\end{equation}
A morphism $f:X\to Y$ is filtration preserving if $f(F^nX)\subset F^nY$ for all $n$.  The resulting category is $\Fil(\V)$.
\end{definition}

\begin{definition}[$N$-truncation and $N$-equivalence]
For $X\in\Fil(\V)$ define
\begin{equation}
\pi_{\leq N}X:=\coker(F^{N+1}X\hookrightarrow X).
\label{eq:truncation}
\end{equation}
A morphism $f:X\to Y$ is an $N$-equivalence, written $f\simeq_{\leq N}$, if $\pi_{\leq N}f$ is an isomorphism.
\end{definition}

The physical interpretation is that $F^nX$ consists of contributions suppressed by at least the $n$th power of a small parameter, typically $\Lambda_{\mathrm{QCD}}/Q$ or a collar radius $\rho$.  Thus $X\simeq_{\leq N}Y$ means equality modulo $\order((\Lambda_{\mathrm{QCD}}/Q)^{N+1})$.

\begin{remark}[Rees presentation]
Equivalently one may use the Rees object
\begin{equation}
\mathrm{Rees}(X)=\bigoplus_{n\geq0}F^nX\,t^n
\end{equation}
and work modulo $t^{N+1}$.  The quotient notation in Eq.~\eqref{eq:truncation} is closer to the standard ``remainder is power suppressed'' formulation, while the Rees notation is often better for derived or graded refinements.
\end{remark}

\subsection{Analytic target and distributional completion}

The algebraic model is not sufficient by itself for QCD.  Coefficient functions, PDFs, Wilson-line matrix elements, plus-distributions, and convolution kernels are distributional objects, and the tensor product used in recomposition is completed.  We therefore distinguish two layers throughout.  The formal layer proves what follows from a balanced collar equivalence.  The analytic layer specifies the functional-analytic and QCD conditions under which the DIS data actually define that equivalence.

\begin{assumption}[Admissible analytic target]
\label{ass:analytic-target}
An analytic target for DIS is a symmetric monoidal category $(\Van,\hot,\mathbf 1)$ satisfying the following conditions.
\begin{enumerate}
\item It is additive and has the strict cokernels and completed coequalizers used to define filtrations and balanced quotients.
\item The completed tensor product $\hot$ preserves those cokernels and coequalizers in the arguments in which it is used, or the required comparison maps are included explicitly in the datum.
\item It contains the distribution spaces needed for DIS: finite channel sums of $\mathcal D'((0,1])$, matrix-valued convolution kernels, renormalized light-ray operator distributions, and their filtered subquotients.
\item Pushforwards along the multiplication map $m:(0,1]\times(0,1]\to(0,1]$ and the DIS measurement maps are continuous morphisms in $\Van$.
\item The filtrations used for power counting are by closed or strict subobjects so that the quotient $X/F^{N+1}X$ is again an object of $\Van$.
\end{enumerate}
\end{assumption}

Examples to keep in mind are suitable quasi-abelian or bornological categories of complete locally convex spaces, or chain complexes therein, equipped with completed projective tensor products.  Nuclearity hypotheses are often useful because they improve the behavior of completed tensor products with distributions~\cite{Treves1967,Hormander1983}.  We do not require a unique analytic model in this paper.  Rather, Assumption~\ref{ass:analytic-target} records the exact functional-analytic verification required to promote the formal theorem from $\Valg$ to the distributional QCD setting.

Thus every theorem below has two readings.  In $\Valg$ it is an ordinary categorical theorem.  In $\Van$ it is conditional on the continuity, completion, and exactness properties listed in Assumption~\ref{ass:analytic-target}.  This separation is intentional: it prevents the formal factorization statement from hiding nontrivial functional analysis.

In particular, the physical balanced core is not merely the algebraic quotient of vector spaces.  At finite accuracy one must verify that the completed diagram
\begin{equation}
\mathsf C_N\hot\mathsf S_N\hot\mathsf F_N
\rightrightarrows
\mathsf C_N\hot\mathsf F_N
\longrightarrow
\mathsf C_N\hot_{\mathsf S_N}\mathsf F_N
\longrightarrow0
\label{eq:strict-completed-core}
\end{equation}
is a strict coequalizer sequence in $\Van$.  Strictness here means that the quotient topology or bornology on the core is the one used by the subsequent convolution and measurement maps.  This requirement is a genuine part of the QCD realization: the coequalizer does not ``sidestep'' the interaction between tensor products, filtrations, and distributions; it isolates the exact completed quotient whose existence and continuity must be checked.

\subsection{Tensor filtrations and finite-accuracy cores}

When two filtered objects are tensored, the physically natural total filtration is
\begin{equation}
F^n(X\otimes Y)=\sum_{a+b\geq n}\im(F^aX\otimes F^bY\to X\otimes Y),
\label{eq:total-filtration}
\end{equation}
or the corresponding completed image for $\hot$ in $\Van$.  This convention implements power counting: a contribution of order $a$ multiplied by one of order $b$ is of order $a+b$.

However, the comparison map between truncating after tensoring and tensoring after truncation is not generally an isomorphism:
\begin{equation}
\pi_{\leq N}(X\otimes Y)
\not\cong
\pi_{\leq N}X\otimes \pi_{\leq N}Y
\quad \text{in general}.
\label{eq:tensor-warning}
\end{equation}
The reason is $\pi_{\leq N}X\otimes \pi_{\leq N}Y$ can retain mixed terms whose total filtration degree exceeds $N$, while $\pi_{\leq N}(X\otimes Y)$ kills them.  The formalism below avoids relying on an unjustified isomorphism of this kind.  The finite-accuracy core is defined either as a coequalizer in the filtered category and then truncated, or directly as a finite-accuracy coequalizer after applying $\pi_{\leq N}$.

\begin{definition}[Finite-accuracy balanced core]
\label{def:finite-core}
Let $\mathsf C,\mathsf S,\mathsf F\in\Fil(\V)$ be a right module, algebra, and left module.  Put $\mathsf C_N=\pi_{\leq N}\mathsf C$, $\mathsf S_N=\pi_{\leq N}\mathsf S$, $\mathsf F_N=\pi_{\leq N}\mathsf F$, and $M_N=(\mathsf C_N\otimes\mathsf S_N)\otimes\mathsf F_N$.  The order-$N$ core may be defined directly in $\V$ by
\begin{equation}
\Core^{\leq N}(\mathsf C,\mathsf S,\mathsf F)
:=\coeq\bigl(M_N \rightrightarrows \mathsf C_N\otimes\mathsf F_N\bigr),
\label{eq:finite-core}
\end{equation}
where the two arrows are $\rho_N\otimes\id$ and $\id\otimes\lambda_N$.
\end{definition}

If $\pi_{\leq N}$ preserves the filtered coequalizer defining $\mathsf C\otimes_{\mathsf S}\mathsf F$, then Eq.~\eqref{eq:finite-core} is canonically isomorphic to $\pi_{\leq N}(\mathsf C\otimes_{\mathsf S}\mathsf F)$.  If this exactness is unavailable in $\Van$, the object in Eq.~\eqref{eq:finite-core} is taken as part of the finite-order DIS datum.  In either case, the balanced core, not the naive tensor product, is the physical carrier.

For a phenomenological reader, a finite-order DIS datum may be read concretely as the following package: a retained channel and operator basis, finite scheme-change kernels acting on that basis, coefficient functions in the dual module, PDF or light-ray correlator data in the collinear module, a power-counting filtration, and the measurement map that sends a coefficient-correlator representative to the observed structure function.  The categorical language does not replace these familiar ingredients; it specifies how they must compose in order to be scheme independent through order $N$.

\section{Interface algebras and balanced cores}
\label{sec:core}

\subsection{Interface algebra}

\begin{definition}[Collinear interface algebra]
At order $N$, the collinear interface algebra is an algebra object
\begin{equation}
\mathsf S^{\leq N}\in \mathrm{Alg}(\Fil(\V))
\end{equation}
encoding all admissible finite collinear counterterms and operator-mixing kernels that preserve the symmetries, quantum numbers, support properties, and power accuracy of the retained DIS sector.
\end{definition}

In $x$-space, $\mathsf S^{\leq N}$ may be modeled by a block matrix algebra of convolution kernels.  In Mellin moment space it becomes a block matrix algebra acting on the retained local twist/operator sector.  The coefficient object $\mathsf C^{\leq N}$ is a right $\mathsf S^{\leq N}$-module and the PDF/correlator object $\mathsf F^{\leq N}$ is a left $\mathsf S^{\leq N}$-module.  The module actions are the categorical version of Eq.~\eqref{eq:module-actions}.

\subsection{Relative tensor product}

\begin{definition}[Balanced core]
Let $\rho:\mathsf C\otimes\mathsf S\to\mathsf C$ be the right action and $\lambda:\mathsf S\otimes\mathsf F\to\mathsf F$ the left action.  The balanced core is the coequalizer
\begin{equation}
\Core:=\mathsf C\otimes_{\mathsf S}\mathsf F
:=\coeq\left(
(\mathsf C\otimes\mathsf S)\otimes\mathsf F
\substack{\xrightarrow{\rho\otimes\id}\\[-0.4em]\xrightarrow[\id\otimes\lambda]{}}
\mathsf C\otimes\mathsf F
\right).
\label{eq:relative-tensor}
\end{equation}
The quotient map is denoted $q:\mathsf C\otimes\mathsf F\to\Core$.
\end{definition}

In a linear category this is the quotient by all elements
\begin{equation}
(c\cdot s)\otimes f-c\otimes(s\cdot f).
\label{eq:balance-linear}
\end{equation}
The defining diagram is
\begin{equation}
\begin{tikzcd}[column sep=large]
(\mathsf C\otimes\mathsf S)\otimes\mathsf F
\arrow[r,shift left=0.45ex,"\rho\otimes\id"]
\arrow[r,shift right=0.45ex,swap,"\id\otimes\lambda"]
& \mathsf C\otimes\mathsf F
\arrow[r,"q"]
& \mathsf C\otimes_{\mathsf S}\mathsf F .
\end{tikzcd}
\label{eq:coeq-diagram}
\end{equation}

\begin{proposition}[Universal property of the core]
\label{prop:balanced-UP}
For every $X\in\V$, composition with $q$ gives a natural bijection
\begin{equation}
\Hom_{\V}(\mathsf C\otimes_{\mathsf S}\mathsf F,X)
\cong
\Hom^{\mathsf S\text{-bal}}_{\V}(\mathsf C\otimes\mathsf F,X),
\label{eq:balanced-UP}
\end{equation}
where the right-hand side consists of morphisms $\varphi$ satisfying
\begin{equation}
\varphi\circ(\rho\otimes\id)=\varphi\circ(\id\otimes\lambda).
\label{eq:balanced-condition}
\end{equation}
\end{proposition}

\begin{proof}
This is the universal property of the coequalizer in Eq.~\eqref{eq:relative-tensor}.
\end{proof}

The same statement can be expressed as a representation theorem for scheme-invariant semantics, following Ref.~\cite{Keller2026Core}.  Let
\begin{equation}
\mathcal B_X(\mathsf C,\mathsf F)=
\{\varphi:\mathsf C\otimes\mathsf F\to X\mid
\varphi(\rho\otimes\id)=\varphi(\id\otimes\lambda)\}
\end{equation}
be the set, or object, of $\mathsf S$-balanced evaluations into $X$.

\begin{proposition}[Core representation and terminal quotient]
\label{prop:core-terminal}
The functor $X\mapsto\mathcal B_X(\mathsf C,\mathsf F)$ is represented by $\mathsf C\otimes_{\mathsf S}\mathsf F$.  Moreover, among quotients $Q$ of the naive composite $\mathsf C\otimes\mathsf F$ through which all balanced evaluations factor, the quotient $q:\mathsf C\otimes\mathsf F\to\mathsf C\otimes_{\mathsf S}\mathsf F$ is terminal: every such quotient maps uniquely to the core whenever it preserves the balanced semantics.
\end{proposition}

\begin{proof}
Representability is Proposition~\ref{prop:balanced-UP}.  Terminality follows because any quotient preserving all balanced evaluations must coequalize the two action maps; hence it receives the universal coequalizer map.  Conversely, every balanced evaluation factors through the core by Proposition~\ref{prop:balanced-UP}.
\end{proof}

The commutative diagram expressing the factorization of a balanced evaluation is
\begin{equation}
\begin{tikzcd}[column sep=large,row sep=large]
\mathsf C\otimes\mathsf F
\arrow[r,"q"]
\arrow[dr,swap,"\varphi"]
& \mathsf C\otimes_{\mathsf S}\mathsf F
\arrow[d,dashed,"\bar\varphi"]\\
& X .
\end{tikzcd}
\label{eq:balanced-evaluation-diagram}
\end{equation}
In DIS, $X$ will be the convolution target $D$, and $\varphi$ will be the pairing of coefficient kernels with PDFs.

\subsection{Truncation of balanced cores}

\begin{proposition}[Comparison of filtered and finite cores]
\label{prop:truncated-core}
Assume coequalizers in $\Fil(\V)$ are computed levelwise and that $\pi_{\leq N}$ preserves the coequalizer in Eq.~\eqref{eq:relative-tensor}.  Then the canonical comparison map is an isomorphism,
\begin{equation}
\pi_{\leq N}(\mathsf C\otimes_{\mathsf S}\mathsf F)
\xrightarrow{\sim}
(\pi_{\leq N}\mathsf C)
\otimes_{\pi_{\leq N}\mathsf S}
(\pi_{\leq N}\mathsf F).
\label{eq:truncated-core}
\end{equation}
\end{proposition}

\begin{proof}
Apply $\pi_{\leq N}$ to the coequalizer diagram defining the relative tensor product.  By hypothesis the resulting diagram remains a coequalizer.  The parallel arrows are the truncated module actions, so the coequalizer is the relative tensor product over $\pi_{\leq N}\mathsf S$.
\end{proof}

\begin{remark}
In the analytic category $\Van$, Proposition~\ref{prop:truncated-core} is not automatic.  The completed coequalizer must be compatible with the chosen closed filtration.  If this is not known a priori, the finite core of Definition~\ref{def:finite-core} is used as part of the datum and the comparison map in Eq.~\eqref{eq:truncated-core} becomes an additional verification condition.
\end{remark}

\subsection{Minimal closure under scheme transformations}

A retained long-distance sector must be stable under the admissible finite counterterms.  The following elementary observation makes this closure explicit.

\begin{definition}[$\mathsf S$-closure]
Let $i:\mathsf F_0\hookrightarrow \mathsf F$ be a subobject of a left $\mathsf S$-module.  Its $\mathsf S$-closure is
\begin{equation}
\langle \mathsf F_0\rangle_{\mathsf S}:=
\im\left(\mathsf S\otimes \mathsf F_0
\xrightarrow{\id\otimes i}\mathsf S\otimes\mathsf F
\xrightarrow{\lambda}\mathsf F\right).
\label{eq:s-closure}
\end{equation}
\end{definition}

\begin{proposition}[Minimal closure]
\label{prop:minimal-closure}
The subobject $\langle \mathsf F_0\rangle_{\mathsf S}$ is the smallest $\mathsf S$-stable subobject of $\mathsf F$ containing $\mathsf F_0$.
\end{proposition}

\begin{proof}
The unit of $\mathsf S$ implies that $\mathsf F_0$ is contained in the image.  Stability follows from associativity of the module action and multiplication in $\mathsf S$.  If $\mathsf G\subset \mathsf F$ is any stable subobject containing $\mathsf F_0$, then the image of $\mathsf S\otimes\mathsf F_0$ under $\lambda$ lies in $\mathsf G$, so $\langle \mathsf F_0\rangle_{\mathsf S}\subset \mathsf G$.
\end{proof}

For leading-twist DIS this reproduces the familiar requirement that singlet quark and gluon operators be treated as a closed mixing block, while nonsinglet sectors close separately.

\section{From pinch surfaces to a stratified regime space}
\label{sec:regime-construction}

This section makes explicit the construction that was only implicit in the preceding discussion.  The stratified regime space is not an additional physical approximation and is not meant to replace the usual list of hard, collinear, soft, jet and overlap regions.  It is the compact bookkeeping object for the incidence data already used in a Collins-style leading-region proof.  Its strata encode region types; its closure relations encode overlaps and further degenerations; and its collar neighborhoods provide the geometric language in which the hard-collinear comparison map is stated.  The order-complex model below is only one convenient realization.  Any conically stratified realization with the same exit-path incidence data is sufficient for the formal theorem.

\subsection{Graph-level scaling data}

Fix a perturbative order and a graph $\Gamma$ contributing to the forward Compton amplitude.  Let $\Lan_\Gamma$ be its real loop-momentum space after imposing momentum conservation, and let $\mathcal D_\Gamma$ be the finite set of denominator factors and numerator tensor structures relevant for power counting.  A pinch singular surface determines which denominators can vanish simultaneously and which components of loop momenta scale relative to the hard scale $Q$.

\begin{definition}[Scaling valuation]
A scaling valuation for $\Gamma$ is a map
\begin{equation}
\nu:\mathcal D_\Gamma\longrightarrow \Q_{\geq0}\cup\{\infty\}
\end{equation}
recording the leading power of a small parameter $\lambda$ in each denominator, numerator, or chosen light-cone component in a sector of loop-momentum space.  Two valuations are equivalent if they give the same leading-power approximator and the same specialization maps to further degenerations.
\end{definition}

For a more concrete model, choose a finite list of scale coordinates
\begin{equation}
r_1,\ldots,r_m
\end{equation}
attached to propagator virtualities and light-cone components after a sector decomposition of $\Lan_\Gamma$.  Writing $r_j\sim \lambda^{u_j}$ produces an exponent vector $u=(u_1,\ldots,u_m)\in \R_{\geq0}^m$.  The leading-power inequalities defining a region are linear inequalities in $u$ together with Landau/pinch compatibility conditions.  Thus a graph-level region may be represented by a rational polyhedral cone
\begin{equation}
C_R\subset \R_{\geq0}^m .
\label{eq:region-cone}
\end{equation}
Faces of $C_R$ correspond to adding further degenerations.  In this cone model,
\begin{equation}
R'\leq R \quad\Longleftrightarrow\quad C_{R'}\text{ is a face or refinement of }C_R,
\label{eq:face-order}
\end{equation}
up to the equivalence that identifies cones inducing the same approximator and specialization maps.  This is the precise sense in which a pinch-surface analysis supplies a stratification: the strata are not arbitrary labels, but equivalence classes of scaling cones.

The use of cones is also the connection with compactification.  Projectivizing a cone by $u\mapsto u/\sum_j u_j$, or equivalently taking the link of the origin in the real oriented blowup of the scaling space, turns rays of asymptotic scaling into boundary directions.  The finite face-poset of these cones is the combinatorial input for the order-complex compactification below.

A valuation is \emph{admissible through order $N$} if the associated power count is not beyond $F^{N+1}$ and if it is compatible with the Landau/pinch equations for the graph.  This gives a finite graph-level poset
\begin{equation}
\Rreg_\Gamma^{\leq N}:=
\{\text{admissible equivalence classes of valuations}\},
\label{eq:graph-region-poset}
\end{equation}
ordered by refinement: $R'\leq R$ if $R'$ is a further degeneration of $R$.  In physical terms, passing from $R$ to $R'$ adds one or more additional scale hierarchies or overlap conditions.

\begin{remark}
At fixed graph and order, finiteness follows from the finite number of denominators, momentum components, and power-counting inequalities.  For an all-orders factorization theorem one passes from graph-specific regions to a finite set of region \emph{types} stable under graph insertion.  The existence and completeness of this finite type system is part of the analytic QCD obligation.
\end{remark}

The graph-level construction is summarized by
\begin{equation}
\begin{tikzcd}[column sep=large]
\Gamma
\arrow[r,"\PSS"]
& \{\nu\}
\arrow[r,"/\sim"]
& \Rreg_\Gamma^{\leq N} .
\end{tikzcd}
\label{eq:graph-to-poset}
\end{equation}

\subsection{Process-level region poset}

For inclusive DIS we collect graph-level region classes into a process-level proof poset $\Rreg_{\DIS}^{\pp}(N)$.  Its elements are region types, not individual graph subgraphs.  We include enough types so that every graph-level admissible valuation maps to one of them:
\begin{equation}
\tau_\Gamma:\Rreg_\Gamma^{\leq N}\longrightarrow \Rreg_{\DIS}^{\pp}(N).
\label{eq:region-type-map}
\end{equation}
The order is again specialization: $R'\leq R$ means that $R'$ lies in the closure of $R$ as a more singular scaling configuration.  The proof-level poset contains the data needed to prove the theorem: hard, incoming-collinear, final-state jet, soft, and overlap types.  The effective leading-power DIS expression has only the hard coefficient and incoming PDF/correlator modules, but that effective two-factor form is obtained only after descent, soft cancellation, and overlap subtraction.

\begin{definition}[Region system at accuracy $N$]
A DIS region system at accuracy $N$ is a finite poset $\Rreg_{\DIS}^{\pp}(N)$ together with maps $\tau_\Gamma$ from all graph-level admissible region posets through order $N$, such that if $R'\leq R$ graph-level, then $\tau_\Gamma(R')\leq \tau_\Gamma(R)$ process-level.
\end{definition}

At fixed loop order $L$ one may make the construction completely finite.  Let $\mathfrak G_L$ be the finite set of renormalized graph topologies contributing to the forward Compton amplitude at that order.  Define a preliminary region category
\begin{equation}
\widetilde{\Rreg}_{\DIS}^{(L,N)}=\bigsqcup_{\Gamma\in\mathfrak G_L}\Rreg_\Gamma^{\leq N}
\end{equation}
and impose the equivalence relation generated by equality of process-level scaling type, equality of the leading approximator, and equality of specialization behavior.  Then
\begin{equation}
\Rreg_{\DIS}^{(L,N)}:=\widetilde{\Rreg}_{\DIS}^{(L,N)}/\sim
\label{eq:LN-region-poset}
\end{equation}
is finite.  A perturbative all-order statement is a compatible family of such finite systems as $L$ varies, together with insertion-stability maps
\begin{equation}
\Rreg_{\DIS}^{(L,N)}\longrightarrow \Rreg_{\DIS}^{(L+1,N)}
\end{equation}
that preserve specialization and power counting.  We mostly suppress $L$ from the notation, but the fixed-$(L,N)$ formulation is the mathematically literal construction.

This definition separates two claims.  The categorical construction requires only a finite poset.  The physical proof must show that the chosen poset is complete for QCD at the desired accuracy and stable under the perturbative operations used in the proof.

\subsection{Order-complex compactification}

\begin{construction}[Order-complex regime space]
\label{cons:order-complex}
Given a finite poset $\Rreg$, define
\begin{equation}
K_{\Rreg}:=|N(\Rreg)|,
\end{equation}
the geometric realization of the nerve of $\Rreg$, equivalently the order complex whose $k$-simplices are chains
\begin{equation}
R_0<R_1<\cdots<R_k .
\end{equation}
The stratum labeled by $R$ is
\begin{equation}
S_R:=\bigcup_{R_0<\cdots<R_j=R}\mathrm{int}(R_0<\cdots<R_j),
\label{eq:stratum-order-complex}
\end{equation}
the union of relative interiors of simplices whose maximal vertex is $R$.
\end{construction}

With the convention that $R'\leq R$ means ``more singular,'' this stratification has the desired incidence relation:
\begin{equation}
S_{R'}\subset \overline{S_R}
\qquad \Longleftrightarrow \qquad
R'\leq R .
\label{eq:incidence}
\end{equation}
Moreover, $K_{\Rreg}$ is a finite PL conically stratified space.  A neighborhood of $S_R$ is modeled on
\begin{equation}
S_R\times \Cone(\Link_R),
\label{eq:conical-chart}
\end{equation}
where $\Link_R$ is the link of the stratum in the order complex.  Thus the incidence data extracted from regions automatically produce conical neighborhoods.

The exit-path category of this finite conically stratified model is equivalent, up to the chosen variance convention, to the specialization category of the poset $\Rreg$.  Consequently, constructible data on $K_{\Rreg}$ can be described either geometrically, by locally constant values on strata and specialization maps across incidences, or combinatorially, as functorial data over $\Rreg$.  This equivalence is the reason the formal theory below can use ordinary diagrams indexed by the region poset while still referring to collars and conical neighborhoods.

\begin{definition}[DIS regime compactification]
For a DIS region system $\Rreg_{\DIS}^{\pp}(N)$, the basic DIS regime compactification is
\begin{equation}
K_{\DIS}^{\pp}(N):=K_{\Rreg_{\DIS}^{\pp}(N)}.
\label{eq:kdis}
\end{equation}
Any conically stratified refinement obtained by adding explicit scaling variables and preserving the same exit-path incidence is called an equivalent DIS regime compactification.
\end{definition}

The order-complex model is minimal.  In a more analytic model, one may attach explicit scale coordinates such as virtuality ratios or light-cone momentum ratios.  A logarithmic or real-oriented blowup compactification then replaces a scale hierarchy by a boundary face.  The formal theorem is invariant under such refinements because it uses only the conical stratification, the collar, and constructible descent.

\subsection{The collinear collar}

Let $S_{\coll}\subset K_{\DIS}$ be the incoming-collinear boundary stratum.  Since $K_{\DIS}$ is conically stratified, $S_{\coll}$ has a conical regular neighborhood.  In the order-complex model this neighborhood can be chosen combinatorially as the open star of $S_{\coll}$,
\begin{equation}
\Star(S_{\coll})=\bigcup_{\sigma\cap S_{\coll}\neq\varnothing}\mathrm{int}(\sigma),
\end{equation}
with transverse directions recorded by the link $\Link(S_{\coll})$.  We choose such a regular neighborhood and denote it by
\begin{equation}
\Ucoll\simeq S_{\coll}\times \Cone(L_{\coll}),
\label{eq:collar-neighborhood}
\end{equation}
where $L_{\coll}$ is the link.  The collar radius is the cone coordinate
\begin{equation}
\rho:\Ucoll\to[0,\epsilon),
\label{eq:rho}
\end{equation}
with $\rho=0$ on $S_{\coll}$.  In a scaling-variable model, $\rho$ may be chosen comparable to the small parameter governing departure from the collinear face.  In the minimal PL model, $\rho$ may be any compatible conical distance function.  Different choices give equivalent filtrations if their ratios are bounded above and below near $S_{\coll}$.

The construction from graph data to collar space is therefore
\begin{align}
\{\text{renormalized DIS graphs}\}
&\longrightarrow \{\Rreg_\Gamma^{\leq N}\}
\longrightarrow \Rreg_{\DIS}^{\pp}(N)
\nonumber\\
&\longrightarrow K_{\DIS}^{\pp}(N)
\longrightarrow (\Ucoll,\rho).
\label{eq:full-regime-pipeline}
\end{align}
The arrows are, respectively, extraction of pinch/scaling data, passage to process-level region types, order-complex realization, and choice of a conical regular neighborhood.  The non-formal inputs are the identification and completeness of the region types; the construction of $K_{\DIS}$ from the finite poset is purely combinatorial.

\section{Constructible regime factorization algebras}
\label{sec:regime-algebra}

\subsection{Disk categories and constructibility}

Let $\Disk(K_{\DIS})$ be a symmetric monoidal category whose objects are finite disjoint unions of sufficiently small conical disks in $K_{\DIS}$ and whose monoidal product is disjoint union.

\begin{definition}[Regime factorization algebra]
A regime prefactorization algebra is a symmetric monoidal functor
\begin{equation}
\A:\Disk(K_{\DIS})\to \Fil(\V).
\end{equation}
A regime factorization algebra is a prefactorization algebra satisfying descent for the covers under consideration.
\end{definition}

\begin{definition}[Constructibility]
A regime factorization algebra $\A$ is constructible if, whenever $D\hookrightarrow D'$ is an inclusion of disks contained in a common stratum, the map $\A(D)\to\A(D')$ is an isomorphism in $\Fil(\V)$.
\end{definition}

Choose disks $D_{\hard}\subset S_{\hard}$ and $D_{\coll}\subset S_{\coll}$.  Constructibility gives well-defined stratum objects
\begin{equation}
\mathsf C:=\A(D_{\hard}),
\qquad
\mathsf F:=\A(D_{\coll}),
\label{eq:stratum-objects}
\end{equation}
up to canonical isomorphism.  These are the hard coefficient object and the collinear PDF/correlator object.

\subsection{Regime algebra from asymptotic expansion}

The notation $\A_{\DIS}$ is meant to package the output of region analysis.  At the level of an individual graph, a conical disk $U\subset K_{\DIS}$ selects a set of region types
\begin{equation}
\Rreg(U)=\{R\in\Rreg_{\DIS}^{\pp}(N)\mid U\cap S_R\neq\varnothing\}.
\end{equation}
A concrete finite-order construction can be described as follows.  Fix loop order $L$ and accuracy $N$.  For each graph $\Gamma\in\mathfrak G_L$, each region $R\in\Rreg(U)$, and each admissible numerator/projection label $\alpha$, introduce a generator
\begin{equation}
[\Gamma,R,\alpha]\quad \text{representing}\quad T_R I_{\Gamma,\alpha},
\end{equation}
where $I_{\Gamma,\alpha}$ is the renormalized integrand or distribution assigned to $\Gamma$ and $T_R$ is the leading-region approximator.  Let $\A_{\mathrm{gen}}^{(L,N)}(U)$ be the completed direct sum of these generators with coefficients in the appropriate distribution spaces of $\Van$, filtered by the power degree assigned by the valuation.  Then $\A_{\DIS}^{(L,N)}(U)$ is obtained from $\A_{\mathrm{gen}}^{(L,N)}(U)$ by imposing the following relations:
\begin{enumerate}
\item renormalization relations among counterterm-subtracted graph distributions;
\item specialization relations for inclusions $V\subset U$ and for incidences $R'\leq R$;
\item nesting relations $T_{R'}T_R\simeq_{\leq N}T_{R'}$ for overlap regions;
\item Ward-identity/eikonalization relations identifying appropriate longitudinal-gluon sums with Wilson-line insertions;
\item scheme-interface relations generated by the action of $\Sch^{\leq N}$ on coefficient and collinear sectors;
\item descent relations for adapted covers of $U$ by conical neighborhoods of strata.
\end{enumerate}
The value $\A_{\DIS}(U)$ is the compatible family of these finite-order objects as $L$ varies, or a fixed-order truncation if one works perturbatively to order $L$.

Restriction maps are induced by forgetting generators whose region labels no longer meet the smaller open set and by applying the specialization maps for deeper strata.  Multiplication is induced by disjoint union of regime disks together with the ordinary product of graph distributions, followed by the same quotient by relations.  Constructibility means local constancy in the regime variable: if two disks lie in the same stratum, they have the same region label and hence canonically isomorphic generator-and-relation presentations.  This does not mean that coefficient functions or PDFs are constant in $x_B$, $Q$, $\mu$, or flavor; those residual analytic and distributional variables live inside the target category $\Van$.

This description is still a construction schema rather than an unconditional theorem.  To make it a theorem for QCD one must prove that the graph-level expansions are compatible with renormalization, gauge identities, subtraction, and the completed tensor structure of $\Van$.  These are precisely the analytic obligations stated in Sec.~\ref{sec:qcd-realization}.

\begin{definition}[Regime extraction datum]
A regime extraction datum at accuracy $N$ consists of
\begin{enumerate}
\item a DIS region system $\Rreg_{\DIS}^{\pp}(N)$ and compactification $K_{\DIS}$;
\item an admissible analytic target $\Van$;
\item for each conical disk $U\subset K_{\DIS}$, a filtered object $\A_{\DIS}(U)\in\Fil(\Van)$ of renormalized asymptotic data supported on the region types meeting $U$;
\item restriction, multiplication, and descent maps making $\A_{\DIS}$ a constructible regime factorization algebra.
\end{enumerate}
\end{definition}

Equivalently, one may view this datum as the value of an asymptotic extraction functor
\begin{equation}
\mathrm{Asymp}_{P,q}:\Fact(M;\Fil(\Van))\longrightarrow \Fact(K_{\DIS};\Fil(\Van))
\label{eq:asymp-functor}
\end{equation}
on the renormalized perturbative observable algebra of QCD.  The functor notation is concise, but the preceding definition makes clear what must be constructed.

\subsection{Balanced collar comparison}

The factorization product places disjoint hard and collinear disks into a collar disk.  This is encoded by a symmetric monoidal embedding functor
\begin{equation}
\iota_{\coll}:\Disk(S_{\hard})\times \Disk(S_{\coll})\to \Disk(\Ucoll).
\end{equation}
Applying $\A$ gives the naive collar map
\begin{equation}
\widetilde\Phi_{\coll}:\mathsf C\otimes\mathsf F\longrightarrow \A(\Ucoll).
\label{eq:naive-collar}
\end{equation}
Scheme invariance requires this map to be $\mathsf S$-balanced:
\begin{equation}
\widetilde\Phi_{\coll}\circ(\rho\otimes\id)
=
\widetilde\Phi_{\coll}\circ(\id\otimes\lambda).
\label{eq:collar-balanced}
\end{equation}
Equivalently, the diagram
\begin{equation}
\begin{tikzcd}[column sep=large,row sep=large]
(\mathsf C\otimes\mathsf S)\otimes\mathsf F
\arrow[r,"\rho\otimes\id"]
\arrow[d,swap,"\id\otimes\lambda"]
& \mathsf C\otimes\mathsf F
\arrow[d,"\widetilde\Phi_{\coll}"]\\
\mathsf C\otimes\mathsf F
\arrow[r,swap,"\widetilde\Phi_{\coll}"]
& \A(\Ucoll)
\end{tikzcd}
\label{eq:balance-square}
\end{equation}
commutes.  By Proposition~\ref{prop:balanced-UP}, $\widetilde\Phi_{\coll}$ descends uniquely to the core:
\begin{equation}
\begin{tikzcd}[column sep=large,row sep=large]
\mathsf C\otimes\mathsf F
\arrow[r,"q"]
\arrow[dr,swap,"\widetilde\Phi_{\coll}"]
& \mathsf C\otimes_{\mathsf S}\mathsf F
\arrow[d,dashed,"\Phi^{\Core}_{\coll}"]\\
& \A(\Ucoll).
\end{tikzcd}
\label{eq:core-collar-diagram}
\end{equation}

\begin{definition}[Order-$N$ scheme-invariant collar factorization]
The regime algebra satisfies order-$N$ scheme-invariant collar factorization if
\begin{equation}
\Phi^{\Core}_{\coll}:
\mathsf C^{\leq N}\otimes_{\mathsf S^{\leq N}}
\mathsf F^{\leq N}\to \A(\Ucoll)
\label{eq:core-collar}
\end{equation}
is an $N$-equivalence.
\end{definition}

This is the central mathematical target of the QCD proof.  It asserts that the full collar observable is equivalent, up to $F^{N+1}$, to the scheme-invariant hard-collinear core.  The effective two-factor form is not assumed at the proof level; it is the result of constructing this balanced collar equivalence after accounting for all proof-level regions.

\section{Measurement and convolution}
\label{sec:measurement}

Let $G=(0,1]$ with multiplication $m(z,\xi)=z\xi$.  Let $D$ be a distributional algebra on $G$ with convolution product
\begin{equation}
 u*v:=m_*(u\boxtimes v).
\label{eq:pushforward-conv}
\end{equation}
For functions this is Eq.~\eqref{eq:conv}.  In $\Van$, $D$ is a completed distribution space and $m_*$ is required to be a continuous morphism.

The DIS measurement consists of Fourier transformation, taking the appropriate discontinuity or cut, Lorentz projection onto the desired structure function, and pushforward to the measured momentum fraction.

\begin{definition}[Disk-constructible kinematic map]
A map $\Pi:K_{\DIS}\to G$ is disk-constructible if inverse images of disks in $G$ decompose into finite disjoint unions of disks in $K_{\DIS}$, inducing a symmetric monoidal pullback functor $\Pi^{-1}:\Disk(G)\to\Disk(K_{\DIS})$.
\end{definition}

\begin{definition}[Measurement]
A measurement of $\A$ with values in $D$ is a symmetric monoidal natural transformation
\begin{equation}
\Meas:\Pi_*\A\Longrightarrow D,
\label{eq:measurement-natural}
\end{equation}
where $D$ is regarded as the constant prefactorization algebra on $G$ with multiplication induced by convolution.
\end{definition}

On the collar, measurement must be compatible with the balanced core.  The relevant diagram is
\begin{equation}
\begin{tikzcd}[column sep=large,row sep=large]
\mathsf C\otimes\mathsf F
\arrow[r,"q"]
\arrow[d,swap,"\widetilde\Phi_{\coll}"]
& \Core
\arrow[d,"\Phi^{\Core}_{\coll}"]
\arrow[ddr,bend left=20,"\Meas_{\Core}"]&\\
\A(\Ucoll)
\arrow[r,"\mathrm{id}"]
& \A(\Ucoll)
\arrow[dr,"\Meas_{\coll}"']&\\
&&D .
\end{tikzcd}
\label{eq:measurement-collar-diagram}
\end{equation}
After applying $\pi_{\leq N}$, the diagonal maps to $D$ are required to agree.  This condition is the categorical version of the equality $\langle \bm C\circledast Z^{-1}, Z\circledast\bm f\rangle=\langle\bm C,\bm f\rangle$.

\section{Internal factorization theorem}
\label{sec:formal-theorem}

We now isolate the purely formal statement.  Before doing so, it is important to state what physical information has already entered.  At fixed perturbative order $L$ and power accuracy $N$, the QCD analysis is assumed to have supplied the retained region types, graph-level approximators, overlap-subtraction maps, Wilson-line reductions to PDF/correlator operators, finite collinear scheme kernels, hard coefficient data, and the DIS measurement projection.  These are the real physical inputs.  The theorem below does not construct them and does not prove the analytic power bound.  It proves that, once these inputs form a balanced datum and the collar map is an $N$-equivalence, the measured factorized expression follows formally.

\begin{definition}[Balanced DIS datum at order $N$]
A balanced DIS datum of order $N$ consists of:
\begin{enumerate}
\item filtered objects $\A_{\mathrm{all}}$, $\A_{\coll}$, $\mathsf C$, and $\mathsf F$;
\item an algebra object $\mathsf S\in\mathrm{Alg}(\Fil(\V))$, a right $\mathsf S$-module structure on $\mathsf C$, and a left $\mathsf S$-module structure on $\mathsf F$;
\item the finite-accuracy core $\Core=\mathsf C\otimes_{\mathsf S}\mathsf F$, or the explicit order-$N$ core of Definition~\ref{def:finite-core};
\item maps in $\Fil(\V)$,
\begin{gather}
\spc:\A_{\mathrm{all}}\to\A_{\coll},\qquad
\Phi:\Core\to\A_{\coll},
\nonumber\\
\bO_{\DIS}:\mathbf 1\to\A_{\mathrm{all}};
\nonumber
\end{gather}
\item a convolution algebra $D\in\V$ and measurement maps
\begin{gather}
\Meas_{\mathrm{all}}:\pi_{\leq N}\A_{\mathrm{all}}\to D,
\nonumber\\
\Meas_{\coll}:\pi_{\leq N}\A_{\coll}\to D,
\nonumber\\
\Meas_{\Core}:\pi_{\leq N}\Core\to D.
\nonumber
\end{gather}
\end{enumerate}
\end{definition}

The datum satisfies the following axioms.

\begin{assumption}[Measurement factors through the collar]
\label{ass:measurement-collar}
\begin{equation}
\Meas_{\mathrm{all}}=
\Meas_{\coll}\circ\pi_{\leq N}(\spc).
\label{eq:meas-factors}
\end{equation}
\end{assumption}

\begin{assumption}[Order-$N$ collar equivalence]
\label{ass:collar-equivalence}
The map $\Phi$ is an $N$-equivalence:
\begin{equation}
\pi_{\leq N}\Phi:\pi_{\leq N}\Core\xrightarrow{\sim}\pi_{\leq N}\A_{\coll}.
\label{eq:collar-equivalence}
\end{equation}
\end{assumption}

\begin{assumption}[Measurement compatibility]
\label{ass:meas-core}
\begin{equation}
\Meas_{\coll}\circ\pi_{\leq N}\Phi=\Meas_{\Core}.
\label{eq:measurement-compatibility}
\end{equation}
\end{assumption}

These data and axioms are summarized by the commutative diagram
\begin{widetext}
\begin{equation}
\begin{tikzcd}[column sep=large,row sep=large]
\mathbf 1
\arrow[r,"\pi_{\leq N}\bO_{\DIS}"]
& \pi_{\leq N}\A_{\mathrm{all}}
\arrow[r,"\pi_{\leq N}\spc"]
\arrow[dr,swap,"\Meas_{\mathrm{all}}"]
& \pi_{\leq N}\A_{\coll}
\arrow[d,"\Meas_{\coll}"]
& \pi_{\leq N}\Core
\arrow[l,swap,"\pi_{\leq N}\Phi", near start]
\arrow[dl,"\Meas_{\Core}"]\\
&&D&
\end{tikzcd}
\label{eq:main-theorem-diagram}
\end{equation}
\end{widetext}
where the left triangle expresses Assumption~\ref{ass:measurement-collar}, the right triangle expresses Assumption~\ref{ass:meas-core}, and the right horizontal arrow is an isomorphism by Assumption~\ref{ass:collar-equivalence}.

\begin{definition}[Order-$N$ factorized expression]
Define
\begin{align}
F_{\DIS}^{(N)}
&:=\Meas_{\Core}\circ(\pi_{\leq N}\Phi)^{-1}
\nonumber\\
&\quad\circ\pi_{\leq N}(\spc)
\nonumber\\
&\quad\circ\pi_{\leq N}(\bO_{\DIS})
:\mathbf 1\to D.
\label{eq:factorized-expression}
\end{align}
\end{definition}

\begin{theorem}[Internal scheme-invariant DIS factorization]
\label{thm:internal-dis}
Given a balanced DIS datum of order $N$ satisfying Assumptions~\ref{ass:measurement-collar}--\ref{ass:meas-core},
\begin{equation}
\Meas_{\mathrm{all}}\circ\pi_{\leq N}(\bO_{\DIS})=F_{\DIS}^{(N)}.
\label{eq:internal-factorization}
\end{equation}
\end{theorem}

\begin{proof}
Use Assumption~\ref{ass:measurement-collar} to rewrite the measured global observable through the collar:
\begin{equation}
\Meas_{\mathrm{all}}\circ\pi_{\leq N}(\bO_{\DIS})
=
\Meas_{\coll}\circ\pi_{\leq N}(\spc)\circ\pi_{\leq N}(\bO_{\DIS}).
\end{equation}
Insert the identity on $\pi_{\leq N}\A_{\coll}$ as
\begin{equation}
\id=
\pi_{\leq N}\Phi\circ(\pi_{\leq N}\Phi)^{-1},
\end{equation}
using Assumption~\ref{ass:collar-equivalence}.  Then apply Assumption~\ref{ass:meas-core}:
\begin{align}
&\Meas_{\coll}\circ\pi_{\leq N}\Phi
\circ(\pi_{\leq N}\Phi)^{-1}
\circ\pi_{\leq N}(\spc)
\circ\pi_{\leq N}(\bO_{\DIS})
\nonumber\\
&\qquad=
\Meas_{\Core}\circ(\pi_{\leq N}\Phi)^{-1}
\circ\pi_{\leq N}(\spc)
\nonumber\\
&\qquad\circ\pi_{\leq N}(\bO_{\DIS}),
\end{align}
which is Eq.~\eqref{eq:factorized-expression}.
\end{proof}

\begin{corollary}[Recovery of the DIS convolution formula]
\label{cor:convolution}
Assume $D$ is the convolution algebra on $(0,1]$, the retained sector decomposes into channels $a\in\mathcal I$, and $\Meas_{\Core}$ is represented by the balanced pairing in Eq.~\eqref{eq:pairing}.  Then Theorem~\ref{thm:internal-dis} gives
\begin{align}
F_i(x_B,Q^2)
&=\sum_{a\in\mathcal I}\int_{x_B}^1\frac{dz}{z}\,
\nonumber\\
&\quad\times C_{i,a}(z,Q,\mu)
\nonumber\\
&\quad\times f_{a/H}\!\left(\frac{x_B}{z},\mu\right)
\nonumber\\
&\quad+\order\!\left((\Lambda_{\mathrm{QCD}}/Q)^{N+1}\right),
\end{align}
where the right-hand side is the image of $\Core^{\leq N}_{\DIS}$, not a scheme-dependent tensor product.
\end{corollary}

\begin{definition}[Valid DIS factorization proof object]
A DIS factorization proof object at order $N$ is valid if the data extracted from Secs.~\ref{sec:regime-construction}--\ref{sec:regime-algebra} produce a balanced DIS datum satisfying Assumptions~\ref{ass:measurement-collar}--\ref{ass:meas-core}.  For a valid object, Theorem~\ref{thm:internal-dis} is the formal factorization theorem.
\end{definition}

This definition is the point at which the framework becomes a method rather than a notation.  A proposed factorization proof can be checked by asking whether every part of the proof object exists and whether the displayed diagrams commute through the claimed filtration order.  For phenomenology, this separates the uncertainty associated with the analytic construction of coefficients and correlators from the consistency conditions that make their recombination scheme independent.  For experimental applications, the measurement map is the typed place where choices of observable, projection, binning or convolution target enter; a failure of measurement compatibility indicates that the proposed prediction space is not yet the right one for the measured quantity.

\section{Diagrammatic realization of the proof object}
\label{sec:diagrammatic-realization}

The formalism above is not intended to replace the usual diagrammatic
factorization proof.  Rather, it encodes the output of such a proof in a
typed form, so that region completeness, overlap subtraction, scheme closure and measurement compatibility become checkable conditions.  This section makes the translation explicit.   A QCD
factorization proof begins with renormalized Feynman graphs, identifies
leading regions and their overlaps, constructs approximators for each
region, subtracts double counting, uses Ward identities to reduce
collinear attachments to Wilson-line operators, and finally proves a
power bound on the remainder.  In the present language, these same steps
construct the objects, maps, and commutative diagrams of the DIS proof
object.

We work at fixed perturbative order $L$ and fixed power accuracy $N$.
Let $\mathfrak G_L$ denote the finite set of renormalized graphs
contributing at this order to the forward Compton amplitude or to its
discontinuity.  For a graph $\Gamma\in\mathfrak G_L$, write
\[
I_\Gamma(\ell_1,\ldots,\ell_L;P,q)
\]
for its renormalized integrand, including numerator, denominators,
color factors, and the measurement/discontinuity prescription when
appropriate.  A diagrammatic proof assigns to $\Gamma$ a finite set of
relevant asymptotic regions
\[
\Rreg_\Gamma^{\leq N}
\]
together with region approximators
\[
T_R I_\Gamma,\qquad R\in \Rreg_\Gamma^{\leq N}.
\]
The region $R$ is the usual pinch or scaling configuration: hard,
target-collinear, soft, jet, or an overlap of such configurations,
depending on the graph and the chosen accuracy.  The formal stratum
$S_R\subset K_\Gamma^{\leq N}$ is simply the typed record of that
scaling configuration.

\begin{table*}[t]
\caption{\label{tab:diagrammatic-realization}
Diagrammatic interpretation of the DIS proof object.  The formal object
does not introduce a new physical approximation; it encodes the standard
diagrammatic ingredients in a typed, checkable form.}
\begin{tabular}{@{}lll@{}}
\toprule
\parbox[t]{0.26\textwidth}{Diagrammatic proof step} &
\parbox[t]{0.33\textwidth}{Formal object} &
\parbox[t]{0.33\textwidth}{Consistency check} \\
\midrule
\parbox[t]{0.26\textwidth}{Renormalized graph integrand $I_\Gamma$} &
\parbox[t]{0.33\textwidth}{Local section or generator of $\A_{\DIS}$} &
\parbox[t]{0.33\textwidth}{Object lies in the chosen analytic target category} \\[0.8ex]

\parbox[t]{0.26\textwidth}{Pinch or leading region $R$} &
\parbox[t]{0.33\textwidth}{Stratum $S_R\subset K_\Gamma^{\leq N}$} &
\parbox[t]{0.33\textwidth}{Region system is complete through order $N$} \\[0.8ex]

\parbox[t]{0.26\textwidth}{Region approximation $T_R I_\Gamma$} &
\parbox[t]{0.33\textwidth}{Specialization to the stratum $S_R$} &
\parbox[t]{0.33\textwidth}{Approximation error lies in $F^{N+1}$ near $S_R$} \\[0.8ex]

\parbox[t]{0.26\textwidth}{Overlap subtraction} &
\parbox[t]{0.33\textwidth}{Descent/M\"obius inversion on $\Rreg_\Gamma^{\leq N}$} &
\parbox[t]{0.33\textwidth}{No double counting of nested regions} \\[0.8ex]

\parbox[t]{0.26\textwidth}{Ward identity and eikonal reduction} &
\parbox[t]{0.33\textwidth}{Module map to Wilson-line/PDF sector $\PDF^{\leq N}$} &
\parbox[t]{0.33\textwidth}{Gauge-compatible reduction to light-ray operators} \\[0.8ex]

\parbox[t]{0.26\textwidth}{Hard coefficient extraction} &
\parbox[t]{0.33\textwidth}{Right $\Sch^{\leq N}$-module $\Coeff^{\leq N}$} &
\parbox[t]{0.33\textwidth}{Coefficient sector transforms contragrediently under schemes} \\[0.8ex]

\parbox[t]{0.26\textwidth}{Finite scheme kernel $Z$} &
\parbox[t]{0.33\textwidth}{Element or morphism of the interface algebra $\Sch^{\leq N}$} &
\parbox[t]{0.33\textwidth}{PDF/correlator sector is closed under $\Sch^{\leq N}$} \\[0.8ex]

\parbox[t]{0.26\textwidth}{Coefficient--PDF recombination} &
\parbox[t]{0.33\textwidth}{Balanced core $\Coeff^{\leq N}\otimes_{\Sch^{\leq N}}\PDF^{\leq N}$} &
\parbox[t]{0.33\textwidth}{Scheme-related presentations define the same core element} \\[0.8ex]

\parbox[t]{0.26\textwidth}{Physical projection to $x_B$} &
\parbox[t]{0.33\textwidth}{Measurement map $\Meas$ to convolution algebra $D$} &
\parbox[t]{0.33\textwidth}{Measurement descends to the balanced quotient} \\[0.8ex]

\parbox[t]{0.26\textwidth}{Power-suppressed remainder} &
\parbox[t]{0.33\textwidth}{Filtration level $F^{N+1}$} &
\parbox[t]{0.33\textwidth}{Exact and factorized constructions agree after $\pi_{\leq N}$} \\
\bottomrule
\end{tabular}
\end{table*}

The most important point is that each diagrammatic operation has a
typed counterpart.  A leading-region analysis is not merely a list of
words such as ``hard'' and ``collinear''; it is a finite incidence
structure.  That incidence structure is the region poset
$\Rreg_\Gamma^{\leq N}$.  Its strata and links define the graph-level
regime space $K_\Gamma^{\leq N}$, while the process-level space
$K_{\DIS}^{(L,N)}$ is obtained by assembling the graph-level region
types and identifying those with the same scaling and approximation
data.

\subsection{Approximators and subtraction as descent}

For each region $R$, the approximator $T_R$ is the usual leading-power
or finite-power expansion appropriate to that region.  For example, in
a hard region it Taylor-expands in small collinear components; in a
collinear region it projects the hard scattering onto the partonic
channel and expands the remaining subgraph around the target-collinear
scaling.  The nesting of regions gives a partial order.  We write
$R'<R$ when $R'$ is a further degeneration or overlap contained in
$R$.

The subtracted region contribution may be written recursively as
\begin{equation}
I^{\mathrm{sub}}_{\Gamma,R}
=
T_R I_\Gamma
-
\sum_{R'<R}
T_R I^{\mathrm{sub}}_{\Gamma,R'} .
\label{eq:diagrammatic-recursive-subtraction}
\end{equation}
This is the diagrammatic subtraction formula in proof-object language.
It says that the contribution assigned to $R$ is its regional
approximation with all more singular contributions removed.  When the
approximators are compatible under nesting, Eq.~\eqref{eq:diagrammatic-recursive-subtraction}
can be written as an incidence-algebra formula
\begin{equation}
I^{\mathrm{sub}}_{\Gamma,R}
=
\sum_{R'\leq R}
\mu_{\Rreg_\Gamma}(R',R)\,
T_{R'} I_\Gamma ,
\label{eq:diagrammatic-mobius-subtraction}
\end{equation}
where $\mu_{\Rreg_\Gamma}$ is the M\"obius function of the region
poset.  Thus the categorical word ``descent'' is not an additional
physical assumption.  It is the formal record of the overlap-subtraction
bookkeeping: local region approximations must agree on common
degenerations after the required subtractions have been made.

The subtracted approximation to the graph in an adapted collar
neighborhood is then
\begin{equation}
I^{\mathrm{fact}}_{\Gamma,U_{\coll}}
=
\sum_{R\in \Rreg_\Gamma(U_{\coll})_{\max}}
I^{\mathrm{sub}}_{\Gamma,R},
\label{eq:graph-level-factorized-approx}
\end{equation}
where the sum is over the maximal retained regions meeting the collar
neighborhood.  The graph-level analytic statement is that
\begin{equation}
I_\Gamma
-
I^{\mathrm{fact}}_{\Gamma,U_{\coll}}
\in
F^{N+1}\A_\Gamma(U_{\coll}).
\label{eq:graph-remainder-filtration}
\end{equation}
Equation~\eqref{eq:graph-remainder-filtration} is the diagrammatic
meaning of equality after applying $\pi_{\leq N}$.

\subsection{Ward identities, Wilson lines, and the core}

The Ward-identity step has a similarly concrete interpretation.  In
the collinear region, longitudinally polarized gluon attachments to the
hard part are replaced, after the usual eikonal approximation and Ward
identity, by Wilson-line attachments in the collinear matrix element.
In the proof object this is the construction of the left
$\Sch^{\leq N}$-module
\[
\PDF^{\leq N},
\]
whose elements are the retained PDF or light-ray operator matrix
elements, and of the right $\Sch^{\leq N}$-module
\[
\Coeff^{\leq N},
\]
whose elements are the corresponding hard coefficient functions.
The interface algebra $\Sch^{\leq N}$ encodes the finite collinear
counterterms, operator-basis transformations, and channel-mixing kernels
that are admissible at the chosen accuracy.

Thus the standard finite scheme transformation
\begin{equation}
f' = Z\circledast f,
\qquad
C' = C\circledast Z^{-1}
\label{eq:diagrammatic-scheme-change}
\end{equation}
is represented by the right and left module actions of
$\Sch^{\leq N}$.  The physical carrier is not the naive tensor product
$\Coeff^{\leq N}\otimes\PDF^{\leq N}$, but the balanced core
\begin{equation}
\Core_{\DIS}^{\leq N}
=
\Coeff^{\leq N}\otimes_{\Sch^{\leq N}}\PDF^{\leq N}.
\label{eq:diagrammatic-core}
\end{equation}
The quotient map
\[
q:\Coeff^{\leq N}\otimes\PDF^{\leq N}
\longrightarrow
\Core_{\DIS}^{\leq N}
\]
imposes the balancing relation
\begin{equation}
q\bigl((C\cdot Z)\otimes f\bigr)
=
q\bigl(C\otimes (Z\cdot f)\bigr).
\label{eq:diagrammatic-balancing}
\end{equation}
Consequently the two scheme presentations in
Eq.~\eqref{eq:diagrammatic-scheme-change} determine the same core
element:
\begin{equation}
q(C'\otimes f')
=
q(C\otimes f).
\label{eq:diagrammatic-scheme-invariant-core}
\end{equation}
This is a concrete diagrammatic check.  If the retained PDF or operator
sector is not closed under the finite kernels $Z$, then
$\PDF^{\leq N}$ is not a left $\Sch^{\leq N}$-module and the balanced
core in Eq.~\eqref{eq:diagrammatic-core} is not well typed.  In
ordinary QCD language, this means that the proposed factorized sector is
missing an operator, channel, or mixing component. 

\subsection{The collar map as the assembled factorized approximation}

The balanced collar comparison map is the formal version of the
assembled subtracted factorized approximation.  Graph by graph, the
subtracted coefficient and collinear matrix-element data determine an
element
\[
\xi_\Gamma\in\Core_{\Gamma}^{\leq N}.
\]
The graph-level collar map sends this core element to the corresponding
subtracted asymptotic representative:
\begin{equation}
\Phi_{\coll,\Gamma}^{\Core}(\xi_\Gamma)
=
I^{\mathrm{fact}}_{\Gamma,U_{\coll}}
\quad
\text{in }
\A_\Gamma(U_{\coll})/F^{N+1}.
\label{eq:graph-collar-map}
\end{equation}
After summing over graphs at order $L$, these maps assemble into
\begin{equation}
\Phi_{\coll}^{\Core,(L)}:
\Core_{\DIS}^{\leq N,(L)}
\longrightarrow
\A_{\DIS}^{(L)}(U_{\coll}).
\label{eq:order-L-collar-map}
\end{equation}
The statement that
$\Phi_{\coll}^{\Core,(L)}$ is an $N$-equivalence is therefore the typed
form of the usual graph-level assertion:
\begin{equation}
\pi_{\leq N}
\left[
I_\Gamma
-
\Phi_{\coll,\Gamma}^{\Core}(\xi_\Gamma)
\right]
=
0
\qquad
\text{for all }\Gamma\in\mathfrak G_L,
\label{eq:diagrammatic-N-equivalence}
\end{equation}
together with the statement that no additional independent leading
classes remain in the collar after imposing the scheme-balancing and
subtraction relations.

This can be summarized as the following criterion.

\begin{proposition}[Diagrammatic criterion for the collar equivalence]
\label{prop:diagrammatic-collar-criterion}
Fix perturbative order $L$ and power accuracy $N$.  Suppose that the
renormalized DIS graphs at order $L$ determine:
\begin{enumerate}
\item a complete region poset $\Rreg_{\DIS}^{(L,N)}$ through order $N$;
\item compatible region approximators $T_R$ satisfying the nesting and
Ward-identity conditions;
\item M\"obius-subtracted graph representatives
$I^{\mathrm{sub}}_{\Gamma,R}$ whose remainders obey
Eq.~\eqref{eq:graph-remainder-filtration};
\item coefficient and PDF/correlator modules
$\Coeff^{\leq N,(L)}$ and $\PDF^{\leq N,(L)}$ over the interface algebra
$\Sch^{\leq N,(L)}$;
\item a balanced graph-level collar map
$\Phi_{\coll}^{\Core,(L)}$.
\end{enumerate}
If every class in
$\pi_{\leq N}\A_{\DIS}^{(L)}(U_{\coll})$ has a representative of the
form
$\pi_{\leq N}\Phi_{\coll}^{\Core,(L)}(\xi)$
and if every $\xi$ mapped to zero is generated by the imposed
scheme-balancing, subtraction, and $F^{N+1}$ relations, then
\[
\Phi_{\coll}^{\Core,(L)}
\]
is an $N$-equivalence.
\end{proposition}

\begin{proof}
After applying $\pi_{\leq N}$, the first condition says that the induced
map
\[
\pi_{\leq N}\Phi_{\coll}^{\Core,(L)}
:
\pi_{\leq N}\Core_{\DIS}^{\leq N,(L)}
\to
\pi_{\leq N}\A_{\DIS}^{(L)}(U_{\coll})
\]
is surjective.  The second condition says that its kernel is zero in the
balanced quotient.  Hence the induced map is an isomorphism in the
finite-accuracy target category.  This is precisely the definition of an
$N$-equivalence.
\end{proof}

\subsection{The gain from the proof-object form}

The diagrammatic proof supplies the analytic content: power counting,
pinch-surface completeness, Ward identities, eikonalization,
renormalization, and control of the remainder.  The proof object records
whether those ingredients close consistently.  In particular, it checks
whether:
\begin{enumerate}
\item all leading regions and their overlaps have been included;
\item the approximators are compatible on nested regions;
\item the Wilson-line/PDF sector is closed under finite scheme kernels;
\item the coefficient sector transforms contragrediently;
\item the measurement map descends to the balanced core; and
\item the exact and factorized collar representatives agree after
$\pi_{\leq N}$.
\end{enumerate}
Thus the proof object is a structured certificate for the diagrammatic
proof.  When it exists, the measured DIS factorization formula follows
formally from the internal theorem.  When it fails to exist, the failure
is localized: a missing stratum indicates an unaccounted leading region,
a failed descent condition indicates an overlap-subtraction defect, a
nonclosed $\Sch^{\leq N}$-module indicates a missing operator or mixing
channel, and a nonbalanced measurement indicates a scheme-dependent
observable presentation.

The same typed structure is useful when some proof components are numerical objects rather than closed-form expressions.  A fitted PDF family, for example, should not be viewed only as a set of functions of $x$; it should define a family of elements in an $\Sch^{\leq N}$-stable module.  Then hard coefficients, fitted long-distance data and the physical measurement compose through the balanced core.  This gives a concrete meaning to ``well typed'' in ML-assisted phenomenology: a learned object must have the module actions and descent maps needed to be compared scheme-invariantly with data.

\section{Using the proof object}
\label{sec:proof-protocol}

The internal theorem is deliberately concise.  Its value is that it separates the analytic proof of factorization from the formal consequences of having completed such a proof.  To make the method explicit, this section spells out the finite-accuracy proof protocol implemented by the DIS example.  A practitioner may use any preferred analytic presentation---diagrammatic Collins analysis, OPE matching, or an effective-theory matching calculation---but the output must populate the same objects and pass the same compatibility checks.

At fixed perturbative order $L$ and fixed power accuracy $N$, the protocol is summarized by
\begin{widetext}
\begin{equation}
\begin{tikzcd}[column sep=small,row sep=large]
\text{DIS graphs}
\arrow[r,"\PSS/\text{scale}"]
& \Rreg_{\DIS}^{(L,N)}
\arrow[r,"|N(-)|"]
& K_{\DIS}^{(L,N)}
\arrow[r,"\text{asympt.}"]
& \A_{\DIS}
\arrow[d,"\text{collar}"]\\
D
& \A_{\DIS}(\Ucoll)
\arrow[l,"\Meas"']
& \Core_{\DIS}^{\leq N}
\arrow[l,"\Phi_{\coll}^{\Core}"']
& (\Coeff,\Sch,\PDF)
\arrow[l,"\coeq"'] .
\end{tikzcd}
\label{eq:proof-protocol-diagram}
\end{equation}
\end{widetext}
Each arrow has a different status.  The passage from a finite region poset to $K_{\DIS}^{(L,N)}$ is combinatorial.  The construction of $\A_{\DIS}$ and $\Phi_{\coll}^{\Core}$ is analytic QCD.  The coequalizer defining the core is categorical.  The final equality of measured observables is formal once the displayed maps satisfy the axioms of Sec.~\ref{sec:formal-theorem}.  An all-order proof is a compatible family of such objects as $L$ varies.

\begin{table*}[t]
\caption{Operational reading of a DIS factorization proof object.  The left column lists the step a practitioner performs, the middle column lists the QCD input, and the right column lists the formal check produced by the present framework.}
\label{tab:proof-protocol}
\renewcommand{\arraystretch}{1.25}
\begin{tabular}{@{}>{\raggedright\arraybackslash}p{0.23\textwidth} >{\raggedright\arraybackslash}p{0.33\textwidth} >{\raggedright\arraybackslash}p{0.36\textwidth}@{}}
\toprule
Protocol step & QCD input & Formal check/output\\
\midrule
Choose accuracy and region system & Fixed perturbative order $L$ and power accuracy $N$; pinch surfaces and scaling valuations & Finite poset $\Rreg_{\DIS}^{(L,N)}$ and type maps $\tau_\Gamma$ for all retained graphs\\[0.8ex]
Compactify regimes & Specialization and overlap incidence of regions & Conically stratified space $K_{\DIS}^{(L,N)}$, exit-path incidence, and collar $(\Ucoll,\rho)$\\[0.8ex]
Construct local asymptotic algebra & Approximators $T_R$, Ward identities, renormalization, and power counting & Constructible filtered $\A_{\DIS}$ with descent and strict/closed filtrations\\[0.8ex]
Extract coefficient and collinear sectors & Hard matching, Wilson-line PDF operators, and finite collinear counterterms & Right/left modules $\Coeff$, $\PDF$ over the interface algebra $\Sch$; minimal $\Sch$-closure\\[0.8ex]
Form invariant carrier & Scheme covariance of coefficients and correlators & Balanced core $\Coeff\otimes_{\Sch}\PDF$ and presentation-independent evaluations\\[0.8ex]
Prove collar equivalence & Power bounds, eikonalization, soft cancellation, and subtraction of overlaps & $\Phi_{\coll}^{\Core}$ is an $N$-equivalence\\[0.8ex]
Measure & DIS projection, discontinuity, Fourier/Mellin transform, and momentum-fraction pushforward & Measurement descends to the core and yields the convolution formula by Theorem~\ref{thm:internal-dis}\\
\bottomrule
\end{tabular}
\end{table*}

\begin{proposition}[Protocol-to-theorem]
\label{prop:protocol-to-theorem}
If the stages in Table~\ref{tab:proof-protocol} are completed in an admissible analytic target and the diagrams defining the core, collar map, descent data, and measurement commute after $N$-truncation, then the resulting data form a valid balanced DIS datum.  Consequently Theorem~\ref{thm:internal-dis} gives the measured order-$N$ factorized expression.
\end{proposition}

\begin{proof}
The region and compactification stages produce $K_{\DIS}$ and $\Ucoll$.  The asymptotic-algebra stage produces $\A_{\DIS}$ and the collar object $\A_{\coll}=\A_{\DIS}(\Ucoll)$.  The scheme stage produces $\Sch$, the modules $\Coeff$ and $\PDF$, and the core.  The collar-equivalence stage is Assumption~\ref{ass:collar-equivalence}.  The measurement stage is Assumptions~\ref{ass:measurement-collar} and~\ref{ass:meas-core}.  These are precisely the hypotheses of Theorem~\ref{thm:internal-dis}.
\end{proof}

\begin{remark}[Checks available before the hard analytic step]
The protocol already performs useful work before the collar equivalence is proved.  It checks that $\Rreg^{\pp}_{\DIS}(L,N)$ is a finite incidence system with a defined collar, that $\PDF$ and $\Coeff$ are actually modules over the same interface algebra, that the naive hard-collinear representative is balanced and descends to the core, that the subtraction data satisfy descent through order $N$, and that the measurement kills the balancing relations.  These are finite structural checks.  The remaining non-formal step is the power-counting and QCD analysis needed to prove that the descended collar map is an $N$-equivalence.
\end{remark}

From the point of view of standard DIS phenomenology, the protocol does not claim that the Ward-identity, power-counting, or soft-cancellation arguments are shorter than in the usual proofs.  Its claim is instead that the analytic proof has a precise target and that several nontrivial consistency properties are checked before the final measurement is taken.  Scheme presentation independence is checked by the core.  Completeness of the retained collinear sector is checked by the existence of the $\Sch$-module action.  Overlap bookkeeping is checked by descent on $\Rreg_{\DIS}^{\pp}$.  Compatibility with the observed structure function is checked by the measurement diagram.  These are the pieces that make the formalism useful as a proof organization and validation tool.

Thus a direct use of the framework for DIS would proceed as follows.  One first fixes $(L,N)$ and writes the graph-level region valuations.  One then constructs the finite region poset and its collar.  Next one defines the local asymptotic generators $[\Gamma,R,\alpha]$ and imposes the relations in Sec.~\ref{sec:regime-algebra}.  Ward identities and Wilson-line reductions identify the collinear stratum object with the chosen PDF sector; matching identifies the hard coefficient module.  Finite counterterms determine $\Sch$.  If the collar map built from the M\"obius-subtracted approximators is an $N$-equivalence, the proof is complete.  No additional manipulation of the final convolution formula is needed, because the scheme-invariant measured expression is the formal image of the core.

\section{A concrete finite check: scheme presentations}
\label{sec:scheme-demo}

The previous theorem reduces analytic factorization to the collar equivalence.  A natural question is what the formalism already checks before one attacks that analytic equivalence.   The simplest useful check is presentation independence under a finite collinear scheme transformation.  We describe it because it is familiar to QCD readers and because it shows why the relative tensor product is not cosmetic.  In the protocol of Sec.~\ref{sec:proof-protocol}, this is a first validation step that can be checked without redoing the full graph-level power-counting proof.

Work in Mellin moment space, where convolution becomes multiplication.  Let $n$ denote the Mellin moment variable, and suppress it from the notation.  In a nonsinglet channel the structure function has a presentation
\begin{equation}
F_2=C_{2,\mathrm{NS}}^{\MS}\, f_{\mathrm{NS}}^{\MS}.
\end{equation}
A DIS-scheme presentation may be obtained through a finite kernel
\begin{equation}
Z_{\DIS\leftarrow\MS}:=C_{2,\mathrm{NS}}^{\MS},
\end{equation}
assuming perturbative invertibility $Z_{\DIS\leftarrow\MS}=1+\order(\alpha_s)$.  Then
\begin{equation}
f_{\mathrm{NS}}^{\DIS}=Z_{\DIS\leftarrow\MS} f_{\mathrm{NS}}^{\MS},
\qquad
C_{2,\mathrm{NS}}^{\DIS}=C_{2,\mathrm{NS}}^{\MS}Z_{\DIS\leftarrow\MS}^{-1}=1,
\label{eq:ns-dis-scheme}
\end{equation}
through the perturbative order being retained.  In a singlet-gluon sector the same statement is matrix-valued: $Z$ is a finite matrix of Mellin-space kernels acting on $(\Sigma,g)^T$, chosen so that the $F_2$ coefficient row takes the conventional DIS-scheme form.  Transformations between $\MS$ and DIS presentations of coefficient functions and parton densities are standard in perturbative DIS calculations~\cite{CataniHautmann1994}.

Categorically, $Z_{\DIS\leftarrow\MS}$ is an invertible element of the interface algebra $\Sch^{\leq N}$.  The two presentations determine two representatives of the same core element:
\begin{align}
q\bigl(C^{\DIS}\otimes f^{\DIS}\bigr)
&=q\bigl((C^{\MS}\cdot Z^{-1})\otimes (Z\cdot f^{\MS})\bigr)
\nonumber\\
&=q\bigl(C^{\MS}\otimes (Z^{-1}Z\cdot f^{\MS})\bigr)
\nonumber\\
&=q\bigl(C^{\MS}\otimes f^{\MS}\bigr).
\label{eq:scheme-demo-core}
\end{align}
The middle equality is precisely the balancing relation defining $\Coeff\otimes_{\Sch}\PDF$.  Thus the equality between $\MS$ and DIS presentations is not imposed after measurement; it already holds in the universal carrier.

\begin{proposition}[Finite scheme covariance is core invariance]
\label{prop:scheme-covariance-core}
Let $z\in\Sch^{\leq N}$ be an invertible finite scheme transformation, let $c\in\Coeff^{\leq N}$ and $f\in\PDF^{\leq N}$, and define $c'=c\cdot z^{-1}$ and $f'=z\cdot f$.  Then
\begin{equation}
q(c'\otimes f')=q(c\otimes f)
\end{equation}
in $\Coeff^{\leq N}\otimes_{\Sch^{\leq N}}\PDF^{\leq N}$.  Consequently every balanced measurement takes the same value on the two scheme presentations.
\end{proposition}

\begin{proof}
The balancing relation gives $q((c\cdot z^{-1})\otimes(z\cdot f))=q(c\otimes(z^{-1}z\cdot f))=q(c\otimes f)$.
\end{proof}

The same point is expressed by the commutative diagram
\begin{equation}
\begin{tikzcd}[column sep=large,row sep=large]
\Coeff^{\MS}\otimes \PDF^{\MS}
\arrow[r,"q"]
\arrow[d,swap,"\mathcal Z"]
& \Coeff\otimes_{\Sch}\PDF
\arrow[d,"\mathrm{id}"]\\
\Coeff^{\DIS}\otimes \PDF^{\DIS}
\arrow[r,"q"]
& \Coeff\otimes_{\Sch}\PDF .
\end{tikzcd}
\label{eq:scheme-change-diagram}
\end{equation}

This finite check also exposes a possible failure mode.  If the retained collinear sector is not closed under $Z$, then $Z\cdot f$ does not lie in the proposed PDF module and the left action of $\Sch$ is not defined.  For instance, keeping a singlet quark distribution while excluding the gluon channel is not stable beyond leading order in a singlet sector.  In the present language this is not a subtle phenomenological caveat; it is a type error in the proof object.  The minimal-closure construction of Proposition~\ref{prop:minimal-closure} repairs the problem by replacing the generating correlators by the smallest $\Sch$-stable sector.

This demonstration does not prove the analytic collar equivalence.  Its role is narrower but important: once the QCD analysis produces any scheme presentation of the hard and collinear factors, the balanced core automatically identifies all finite scheme-related presentations and tests whether the chosen operator sector was large enough.  It is therefore a concrete check that the proof object is doing work before the final factorization equivalence is invoked.  In a data-analysis workflow, this check can be applied to a fitted model class before the final observable is formed: a parametrization that cannot be acted on by the retained finite kernels is not a valid parametrization of the chosen factorized sector.

\section{Subtractions as descent on the region poset}
\label{sec:descent}

Theorem~\ref{thm:internal-dis} is formal.  The analytic problem is the construction of a datum satisfying its assumptions, especially the collar equivalence.  This section explains how Collins-style overlap subtraction appears as descent over $\Rreg^{\pp}_{\DIS}$.

Let $\{U_R\}_{R\in\Rreg}$ be an adapted cover of $K_{\DIS}$ by neighborhoods of strata.  Intersections correspond to common degenerations: $U_{R_0}\cap\cdots\cap U_{R_k}$ is nonempty only when the corresponding regions admit a common nested refinement.  Descent gives a Cech complex of the schematic form
\begin{equation}
\begin{aligned}
0\to \A(K_{\DIS})&\to \prod_R\A(U_R)\\
&\to \prod_{R_0<R_1}\A(U_{R_0}\cap U_{R_1})\to\cdots .
\end{aligned}
\label{eq:cech}
\end{equation}

For each region $R$, let
\begin{equation}
T_R:\A(U)\to \A(U_R)
\end{equation}
be the corresponding approximator.  Depending on $R$, this may be a hard Taylor expansion, a collinear expansion, a soft approximation, an eikonal approximation, or an iterated approximation in an overlap.  The nesting condition is that for $R'\leq R$, the two composites to the deeper region agree through order $N$:
\begin{equation}
T_{R'}T_R\simeq_{\leq N}T_{R'}.
\label{eq:nesting}
\end{equation}
Diagrammatically,
\begin{equation}
\begin{tikzcd}[column sep=large,row sep=large]
\A(U)
\arrow[r,"T_R"]
\arrow[dr,swap,"T_{R'}"]
& \A(U_R)
\arrow[d,"T_{R'R}"]\\
& \A(U_{R'})
\end{tikzcd}
\qquad \text{commutes after }\pi_{\leq N}.
\label{eq:nesting-diagram}
\end{equation}

Let $\zeta(R',R)=1$ when $R'\leq R$ and $0$ otherwise, and let $\mu_{\Rreg}$ be the M\"obius function of the incidence algebra.  If $A_R$ denotes the naive approximation assigned to region $R$, define subtracted contributions
\begin{equation}
A_R^{\mathrm{sub}}=
\sum_{R'\leq R}\mu_{\Rreg}(R',R)A_{R'}.
\label{eq:mobius-sub}
\end{equation}
Equivalently, the naive approximations are recovered by zeta summation:
\begin{equation}
A_R=\sum_{R'\leq R}A_{R'}^{\mathrm{sub}}.
\end{equation}

\begin{proposition}[Subtraction/descent principle]
\label{prop:subtraction-descent}
Assume that $\A$ satisfies descent for the adapted cover and that the region approximators obey the nesting condition \eqref{eq:nesting}.  Then the M\"obius-subtracted family $\{A_R^{\mathrm{sub}}\}$ glues to the full collar observable through order $N$.  Equivalently, the Cech obstruction to gluing lies in $F^{N+1}$.
\end{proposition}

\begin{proof}
The Cech differential measures disagreement of local data on intersections.  By construction of the adapted cover, intersections are indexed by chains in the region poset.  The nesting condition says that the restrictions of the approximators to each chain agree after applying $\pi_{\leq N}$.  The zeta transform records the accumulation of contributions from nested regions, and M\"obius inversion is its inverse in the incidence algebra.  Therefore Eq.~\eqref{eq:mobius-sub} removes every overlap with the coefficient required by the Cech alternating sum.  All remaining descent defects lie in $F^{N+1}$ and vanish under $\pi_{\leq N}$.
\end{proof}

\begin{example}[Two-region overlap]
\label{ex:two-region-overlap}
Suppose an open collar is covered by two leading regions $H$ and $C$ with an overlap $O=H\cap C$.  The proof-level poset has $O\leq H$ and $O\leq C$.  The naive sum $A_H+A_C$ double counts the overlap approximation $A_O$.  The M\"obius function gives
\begin{equation}
A_H^{\mathrm{sub}}=A_H-A_O,
\qquad
A_C^{\mathrm{sub}}=A_C-A_O,
\qquad
A_O^{\mathrm{sub}}=A_O,
\end{equation}
or, for the maximal-region reconstruction,
\begin{equation}
A_{H\cup C}^{\mathrm{sub}}=A_H+A_C-A_O .
\label{eq:two-region-subtraction}
\end{equation}
This is the familiar subtraction of a common limiting approximation.  In the present language Eq.~\eqref{eq:two-region-subtraction} is the degree-one Cech gluing condition for the cover $\{U_H,U_C\}$.  More complicated Collins subtraction forests are the same incidence-algebra calculation on a larger region poset.
\end{example}

In a graph-by-graph proof, the local data are expanded integrands or renormalized graph contributions.  In the present formulation, the same inclusion-exclusion structure is a local-to-global descent statement on the stratified regime space.  This does not remove the need to prove that the approximators are valid; it makes the overlap bookkeeping a finite combinatorial problem once the region system is known.

\section{QCD realization and the collar obligation}
\label{sec:qcd-realization}

\subsection{Perturbative observable algebra and asymptotic extraction}

Let $M=\R^{1,3}$.  We write
\begin{equation}
\mathrm{Obs}^{\mathrm{ren,pert}}_{\mathrm{QCD}}\in \Fact(M;\Fil(\Van))
\end{equation}
for the renormalized perturbative gauge-invariant observable algebra of QCD.  This notation is not a claim of a nonperturbative construction of four-dimensional QCD as a factorization algebra; it denotes the perturbative, renormalized, gauge/BRST-compatible object from which current insertions and Wilson-line light-ray operators are drawn.

For fixed external kinematics $(P,q)$ the regime extraction datum of Sec.~\ref{sec:regime-algebra} may be summarized as an asymptotic extraction functor
\begin{equation}
\mathrm{Asymp}_{P,q}:\Fact(M;\Fil(\Van))\to \Fact(K_{\DIS};\Fil(\Van))
\label{eq:asymp}
\end{equation}
whose value on QCD observables is
\begin{equation}
\A_{\DIS}:=\mathrm{Asymp}_{P,q}(\mathrm{Obs}^{\mathrm{ren,pert}}_{\mathrm{QCD}}).
\end{equation}
Operationally, this packages the leading-region expansion, the specializations to strata, the completed distributional tensor products, and the power-counting filtration.

\subsection{Constructing the collar map from graph data}
\label{subsec:construct-collar}

At fixed $(L,N)$ the balanced collar map can be described on generators before passing to the completed quotient.  A renormalized graph contribution $I_\Gamma$ has region approximations $T_R I_\Gamma$.  The hard stratum projection extracts a coefficient generator $C_{\Gamma,R_h}$, the incoming-collinear stratum projection extracts a correlator generator $F_{\Gamma,R_c}$, and the factorization product gives a hard-collinear representative
\begin{equation}
C_{\Gamma,R_h}\otimes F_{\Gamma,R_c}
\in \Coeff_N\otimes \PDF_N .
\end{equation}
The descent-corrected collar representative is the M\"obius-subtracted family
\begin{equation}
I^{\mathrm{sub}}_{\Gamma,\Ucoll}
=
\sum_{R\in\Rreg(\Ucoll)}\mu_{\Rreg}(R,\Ucoll)\,T_R I_\Gamma,
\label{eq:graph-collar-sub}
\end{equation}
where the notation $\mu_{\Rreg}(R,\Ucoll)$ abbreviates the coefficient determined by the incidence algebra of the adapted collar cover.  The graph-level collar assignment is therefore the composite
\begin{equation}
\begin{tikzcd}[column sep=large]
\Coeff_N\otimes\PDF_N
\arrow[r,"q"]
& \Coeff_N\otimes_{\Sch_N}\PDF_N
\arrow[r,"\Phi^{\Core}_{\coll}"]
& \A_{\DIS}(\Ucoll)_N,
\end{tikzcd}
\end{equation}
with the last arrow sending a balanced class of hard-collinear generators to the class of Eq.~\eqref{eq:graph-collar-sub}.  The map is well defined only after three verifications: the Ward identities identify equivalent longitudinal-gluon representatives, the scheme relations are balanced over $\Sch_N$, and the overlap-subtracted family satisfies descent through order $N$.

This description is deliberately more explicit than the functor notation in Eq.~\eqref{eq:asymp}.  It identifies the algebraic target of the traditional proof: prove that every renormalized collar contribution is represented uniquely modulo $F^{N+1}$ by a balanced hard-collinear class, and prove that the construction is continuous in the analytic target.  Those are precisely the obligations listed next.

\subsection{Analytic QCD obligations}

The following assumptions are the analytic content needed to prove the collar equivalence.  They are stated as obligations because the purpose of this paper is to formalize the target of the proof, not to reproduce all graph-level estimates.

\begin{assumption}[Complete region system]
\label{ass:leading-regions}
For inclusive DIS at fixed $x_B\in(0,1)$ away from endpoints, the proof-level region poset $\Rreg^{\pp}_{\DIS}(N)$ contains every pinch/leading region type contributing through order $N$.  Every graph-level admissible valuation maps to this poset as in Eq.~\eqref{eq:region-type-map}, and all omitted regions contribute to $F^{N+1}$ after renormalization and measurement.
\end{assumption}

\begin{assumption}[Admissible analytic realization]
\label{ass:admissible-realization}
The QCD asymptotic data live in an analytic target $\Van$ satisfying Assumption~\ref{ass:analytic-target}.  The completed tensor products, distributional pushforwards, cokernels, coequalizers, and filtration quotients used to define $\A_{\DIS}$, $\Sch$, $\Coeff$, $\PDF$, and $D$ exist and are compatible with the required maps.
\end{assumption}

\begin{assumption}[Ward-identity eikonalization]
\label{ass:ward}
The longitudinal gluon attachments connecting the hard subgraph to the incoming-collinear subgraph can be summed by Ward identities into Wilson lines along the incoming lightlike direction.  The resulting collinear matrix element is the gauge-invariant PDF/correlator object $\mathsf F$.
\end{assumption}

\begin{assumption}[Inclusive soft cancellation]
\label{ass:soft-cancel}
After summing over final states and applying the inclusive DIS measurement, leading soft exchanges either cancel or are absorbed into the Wilson-line definition of $\mathsf F$, so no independent soft module survives in the effective DIS core.
\end{assumption}

\begin{assumption}[Overlap subtraction and power bound]
\label{ass:overlap}
The region approximators satisfy the nesting condition \eqref{eq:nesting}, their M\"obius-subtracted sum reconstructs the full collar observable through order $N$, and the residual error lies in $F^{N+1}$.
\end{assumption}

\begin{assumption}[Scheme-closed collinear sector]
\label{ass:scheme-closed}
The retained PDF/operator sector is stable under the finite collinear counterterms and operator-mixing kernels at the chosen accuracy.  Equivalently, it is an $\mathsf S^{\leq N}$-module, or it has been replaced by the minimal closure of Proposition~\ref{prop:minimal-closure}.
\end{assumption}

\begin{proposition}[QCD obligations imply the collar equivalence]
\label{prop:qcd-collar}
Assumptions~\ref{ass:leading-regions}--\ref{ass:scheme-closed} imply that the QCD regime algebra admits a balanced collar comparison map
\begin{equation}
\Phi^{\Core}_{\coll}:
\mathsf C^{\leq N}\otimes_{\mathsf S^{\leq N}}
\mathsf F^{\leq N}\to \A_{\DIS}(\Ucoll)
\end{equation}
which is an $N$-equivalence.
\end{proposition}

\begin{proof}[Proof strategy]
Assumption~\ref{ass:leading-regions} identifies the complete proof-level cover of the relevant collar and gives the finite stratified space of Sec.~\ref{sec:regime-construction}.  Assumption~\ref{ass:admissible-realization} ensures that all distributional and completed monoidal operations are legitimate in $\Van$.  Assumption~\ref{ass:ward} turns the collinear attachments into the Wilson-line PDF object.  Assumption~\ref{ass:soft-cancel} removes independent soft factors from the effective inclusive DIS carrier.  Assumption~\ref{ass:overlap}, together with Proposition~\ref{prop:subtraction-descent}, says that the M\"obius-subtracted hard-collinear contribution glues to the full collar observable after $N$-truncation.  Assumption~\ref{ass:scheme-closed} ensures that finite collinear scheme transformations act internally on the retained sector.  Therefore the naive hard-collinear collar map is balanced over $\mathsf S^{\leq N}$ and descends to the core.  The glued reconstruction gives an inverse after applying $\pi_{\leq N}$.
\end{proof}

Combining Criterion~\ref{prop:qcd-collar} with Theorem~\ref{thm:internal-dis} gives the formal shape of the DIS factorization theorem: analytic QCD proves the collar $N$-equivalence; filtered monoidal algebra turns it into the measured convolution formula.

\subsection{Pipeline diagram}

The division of labor can be summarized by
\begin{widetext}
\begin{equation}
\begin{tikzcd}[column sep=large,row sep=large]
\mathrm{Obs}^{\mathrm{ren,pert}}_{\mathrm{QCD}}
\arrow[r,"\mathrm{Asymp}_{P,q}"]
& \A_{\DIS}\in\Fact(K_{\DIS};\Fil(\Van))
\arrow[r,"\text{collar}"]
& \A_{\DIS}(\Ucoll)
& \Core^{\leq N}_{\DIS}
\arrow[l,swap,"\Phi^{\Core}_{\coll}", "\sim_{\leq N}"]
\arrow[r,"\Meas_{\Core}"]
& D .
\end{tikzcd}
\label{eq:pipeline}
\end{equation}
\end{widetext}
The only non-formal arrow in the factorization theorem is the proof that $\Phi^{\Core}_{\coll}$ is an $N$-equivalence for QCD in the chosen analytic target.

\section{Discussion}
\label{sec:discussion}

\subsection{Formal versus analytic}

The categorical theorem proves that a balanced collar equivalence implies a scheme-invariant factorized expression after measurement.  It does not by itself prove leading-region completeness, Ward identities, soft cancellation, or power suppression.  Those remain analytic QCD tasks.  The benefit is that all those tasks are aimed at a single precise statement,
\begin{equation}
\Core^{\leq N}_{\DIS}\simeq_{\leq N}\A_{\DIS}(\Ucoll),
\end{equation}
rather than at a presentation-dependent tensor product.  The proof object also records whether the proposed factorization is closed under finite scheme transformations and whether the measurement descends to the balanced quotient.

For inclusive DIS this formalism is not advertised as a shorter replacement for the established all-orders analytic proof.  DIS is deliberately used as a base case because its physics is clean enough that the proof object can be displayed.  The concrete advantages shown here are structural: Sec.~\ref{sec:proof-protocol} turns the proof into an explicit protocol, Sec.~\ref{sec:scheme-demo} proves presentation independence before measurement and exposes nonclosed parton sectors as invalid data, while Sec.~\ref{sec:descent} turns overlap subtraction into an incidence-algebra computation.  These are modest but nontrivial checks in the simplest setting; the larger purpose is to define the proof calculus in a case where QCD readers can compare every object with familiar DIS ingredients.

\subsection{A phenomenological reading}

A phenomenologist normally begins with a calculational presentation: a choice of factorization scheme, operator basis, coefficient functions, PDFs, evolution kernels, and a measurement projection.  The formalism proposed here changes the validation order.  Rather than first forming the measured convolution and then checking that scheme-dependent pieces cancel, one first lifts the presentation to the balanced core and only then measures:
\begin{equation}
(\bm C,\bm f,\Sch)\longrightarrow
\bm C\otimes_{\Sch}\bm f
\longrightarrow D .
\label{eq:phenom-order}
\end{equation}
In this order, scheme covariance, sector closure, overlap subtraction, and measurement compatibility are separate checks.  This is useful even when all physical ingredients are already known, because it prevents distinct issues from being hidden inside the same final convolution formula.

The framework should therefore be evaluated as proof infrastructure rather than as a new calculation of DIS coefficient functions.  A Collins-style proof, an OPE matching calculation, or an effective-theory matching calculation can all provide a presentation of the same proof object.  Once the proof object is valid, Theorem~\ref{thm:internal-dis} identifies the scheme-invariant measured observable without rechecking presentation-dependent cancellations case by case.  This is the sense in which the method is meant to be direct and seamless for factorization proofs: QCD supplies the typed collar equivalence, while the formalism supplies the invariant recomposition and consistency checks.

\subsection{Typed parameterizations and PDF extraction}
\label{subsec:typed-pdf-extraction}

The same validation order is useful for modern PDF extractions.  In neural-network approaches to PDFs, the fit parameters specify a family of functions or distributions that are then passed through evolution, coefficient functions and observable maps~\cite{ForteGarridoLatorrePiccione2002,DelDebbioEtAl2007NNPDF,BallEtAl2009NNPDF,ForteCarrazza2020,CarrazzaCruzMartinez2019,BallEtAl2022NNPDF40}.  In the present language such a family should be typed as a map
\begin{equation}
\theta\longmapsto f_\theta\in \PDF^{\leq N}
\end{equation}
whose image is stable under the interface algebra $\Sch^{\leq N}$.  The physical prediction is then not the unbalanced expression $C\otimes f_\theta$, but the composite
\begin{equation}
\theta\longmapsto
\Meas_{\Core}\bigl([C\otimes f_\theta]\bigr)
\in D .
\end{equation}
This formulation makes scheme closure, basis changes and measurement compatibility part of the model specification rather than after-the-fact checks on a fitted curve.

A neural architecture that fails to land in an $\Sch^{\leq N}$-stable module is not merely a suboptimal parametrization; it is not well typed for the factorization problem.  Conversely, an architecture or loss function that respects the right and left module actions, the balancing relation and the measurement map represents a scheme-invariant prediction space from the start.  Categorical deep learning provides a language for such compositional architectures, because it treats model components as typed maps whose algebraic constraints can be specified independently of a particular numerical implementation~\cite{GavranovicEtAl2024CategoricalDL}.  The proof-object perspective therefore suggests concrete validation tests for ML-based phenomenology: check module closure of the fitted family, check that scheme-related presentations descend to the same core element, and compare data only after the physical measurement has factored through the balanced core.

\subsection{Relation to standard presentations}

CSS-style proofs, the OPE, and effective-theory matchings provide explicit presentations of the same data: operator bases, matching coefficients, renormalization kernels, leading-region approximators, and scale evolution.  In the present formalism, such a presentation supplies a model for $\mathsf C$, $\mathsf F$, and $\mathsf S$, as well as a candidate for the collar map.  Changing scheme or basis changes the presentation but not the core.  This makes the framework compatible with standard calculations while separating presentation-dependent choices from invariant content.

\subsection{Why the stratified formulation is useful}

The stratified regime space is not introduced merely to rename the set of regions.  It gives a local-to-global structure to the proof.  Missing regions appear as missing strata.  Double counting appears as nontrivial overlap data.  Subtraction appears as descent.  Power accuracy appears as a filtration.  Scheme dependence appears as an interface algebra.  These are independent consistency checks on a proposed factorization proof object.  In particular, the framework distinguishes failures that can look similar in a final formula: failure of region completeness, failure of scheme closure, failure of descent/subtraction, failure of measurement compatibility, and failure of the collar power estimate.

\subsection{Obstruction diagnostics and diagnostic value}
\label{subsec:obstruction-diagnostics}

The present paper does not claim a new inclusive-DIS prediction.  Its value is diagnostic.  A failed attempt to construct a proof object does not merely say ``factorization failed''; it localizes the obstruction to a typed part of the construction.  Table~\ref{tab:obstructions} summarizes the main cases.

\begin{table*}[t]
\caption{Formal obstructions and their QCD interpretation.  The table is meant as a diagnostic guide for proposed factorization proofs, not as a list of new DIS effects.}
\label{tab:obstructions}
\begin{tabular}{@{}llll@{}}
\toprule
\parbox[t]{0.22\textwidth}{Formal obstruction} &
\parbox[t]{0.27\textwidth}{Mathematical signal} &
\parbox[t]{0.31\textwidth}{QCD interpretation} &
\parbox[t]{0.14\textwidth}{Typical repair}\\
\midrule
\parbox[t]{0.22\textwidth}{Incomplete region system} &
\parbox[t]{0.27\textwidth}{Some graph-level valuation has no image in $\Rreg^{\pp}_{\DIS}(N)$, or an omitted region contributes outside $F^{N+1}$.} &
\parbox[t]{0.31\textwidth}{Missing leading region, overlap face, or boundary contribution in the asymptotic analysis.} &
\parbox[t]{0.14\textwidth}{Enlarge $\Rreg$ and $K$.}\\[0.5ex]
\parbox[t]{0.22\textwidth}{Bad analytic target} &
\parbox[t]{0.27\textwidth}{Completed tensor products, strict quotients, or convolution pushforwards do not exist or are not continuous.} &
\parbox[t]{0.31\textwidth}{Distributional products or renormalized kernels have not been placed in a category where the formal operations are legal.} &
\parbox[t]{0.14\textwidth}{Refine $\Van$ or topology.}\\[0.5ex]
\parbox[t]{0.22\textwidth}{Nonclosed collinear sector} &
\parbox[t]{0.27\textwidth}{$\Sch^{\leq N}\cdot\PDF^{\leq N}\nsubseteq\PDF^{\leq N}$.} &
\parbox[t]{0.31\textwidth}{Missing PDF/operator channel or operator-mixing block.} &
\parbox[t]{0.14\textwidth}{Apply minimal closure.}\\[0.5ex]
\parbox[t]{0.22\textwidth}{Nonbalanced collar map} &
\parbox[t]{0.27\textwidth}{The square in Eq.~\eqref{eq:balance-square} fails.} &
\parbox[t]{0.31\textwidth}{The hard-collinear presentation depends on a finite scheme choice or counterterm convention.} &
\parbox[t]{0.14\textwidth}{Correct scheme action or module data.}\\[0.5ex]
\parbox[t]{0.22\textwidth}{Descent defect} &
\parbox[t]{0.27\textwidth}{The Cech obstruction or M\"obius-subtraction remainder survives after $\pi_{\leq N}$.} &
\parbox[t]{0.31\textwidth}{Overlap subtraction is incomplete or incompatible with nesting of approximators.} &
\parbox[t]{0.14\textwidth}{Refine subtraction poset.}\\[0.5ex]
\parbox[t]{0.22\textwidth}{Failed collar equivalence} &
\parbox[t]{0.27\textwidth}{Kernel or cokernel of $\pi_{\leq N}\Phi^{\Core}_{\coll}$ is nonzero.} &
\parbox[t]{0.31\textwidth}{An uncancelled or unabsorbed contribution remains at the claimed power.} &
\parbox[t]{0.14\textwidth}{Identify missing cancellation, region, or module.}\\[0.5ex]
\parbox[t]{0.22\textwidth}{Nonbalanced measurement} &
\parbox[t]{0.27\textwidth}{$\Meas$ does not factor through $\Core$.} &
\parbox[t]{0.31\textwidth}{The measured expression is scheme-dependent or the convolution target is incomplete.} &
\parbox[t]{0.14\textwidth}{Modify measurement target or projection.}\\
\bottomrule
\end{tabular}
\end{table*}

This diagnostic interpretation is the conservative way in which the framework can lead to new physics insight.  It does not predict a new contribution in the base DIS setting.  Rather, it gives a disciplined way to discover when an attempted factorization proof is missing physical data: a region, an operator sector, a subtraction, a cancellation, or a measurement variable.  In a proof assistant or AI-guided workflow, the same diagnostics become type errors, failed commutative diagrams, or nonvanishing finite-accuracy obstructions.

\subsection{Future refinements}

The present paper intentionally treats inclusive DIS away from special kinematic boundaries as the base case.  Future work can enlarge the same architecture by adding boundary faces, additional measurement variables, or extra module factors.  In this language such extensions are not changes to the formal theorem; they are changes to the region poset, the interface algebra, and the measurement functor.  The base DIS construction is therefore meant to be the first verified instance of a more general proof calculus, but the details of those enlargements are outside the scope of this paper.

A second direction is computational.  On the proof side, the finite list of regions, modules, maps, coequalizers and diagrams is suitable for proof-assistant formalization, with Lean used here only for a finite type-level model of the internal theorem~\cite{DeMouraUllrich2021Lean4,Mathlib2020}.  On the phenomenology side, the same typed interface can guide neural PDF extractions, differentiable implementations of convolution maps, and automated searches for missing sectors or inconsistent scheme actions.  In more complicated processes, a failed proof object could become a systematic tool for proposing which new region, operator or measurement variable must be added before a factorized description can close.

\section{Conclusion}
\label{sec:conclusion}

We have formulated inclusive DIS factorization as a scheme-invariant collar theorem for a constructible filtered factorization algebra on a stratified regime space.  The stratification is constructed from pinch-surface and scaling data by passing to a finite region poset and then to an order-complex compactification.  The interface algebra encodes finite collinear scheme transformations, and the physical carrier of the factorized expression is the balanced core, not the naive tensor product of a coefficient presentation and a PDF presentation.

The main formal result is concise.  If the balanced collar comparison map from the core to the collar value of the regime algebra is an $N$-equivalence, and if measurement is convolution-compatible, then the measured DIS observable is the measured image of the core modulo the power-suppressed filtration.  The usual $x$-space convolution formula is recovered as a representation of this invariant statement.  Collins-style subtraction appears as descent and M\"obius inversion on the proof-level region poset.

The framework gives a precise division of labor.  Analytic QCD supplies the collar equivalence by proving region completeness, Ward-identity reduction to Wilson lines, inclusive soft cancellation, overlap subtraction and power bounds.  The formalism supplies the invariant recomposition, finite scheme covariance, measurement descent and obstruction diagnostics.

The same structure also suggests a typed interface for computational workflows.  Proof assistants can check the finite diagrammatic core once the analytic obligations are stated, while neural-network PDF parametrizations can be constrained to land in scheme-stable modules and to be compared only after descent to the balanced core.  Thus the immediate output is a proof protocol for DIS, and the broader goal is a reusable calculus for validating factorization proofs, diagnosing missing physics and organizing hybrid analytic--computational extractions.

\acknowledgments

The author acknowledges the assistance of Overleaf AI (AI Assist, powered by Writefull) for language editing, including improvements to grammar, clarity, and flow of the manuscript.  The Lean theorem prover (Lean 4) was used to check a finite type-level model of the balanced-core construction and the internal diagram-chase theorem described in Appendix~\ref{app:lean}; this formalization does not constitute a verification of the analytic QCD collar-equivalence assumptions.  All outputs from both tools were reviewed, edited and verified by the author.
\bibliographystyle{JHEP}
\bibliography{references}

\appendix

\section{Homotopy-coherent and type-theoretic refinement}
\label{app:homotopy}

The main text uses an ordinary coequalizer because that is the clearest finite-accuracy model for inclusive DIS and because it gives a direct theorem in a 1-categorical proof assistant setting.  In a derived or homotopy-coherent target, for example chain complexes of distributional spaces or BRST complexes, the relative tensor product should be replaced by the derived balanced core
\begin{equation}
\mathsf C\otimes^{\mathbb L}_{\mathsf S}\mathsf F
:=\left|B_\bullet(\mathsf C,\mathsf S,\mathsf F)\right|,
\label{eq:derived-core}
\end{equation}
where $B_\bullet$ is the two-sided bar construction with
\begin{equation}
B_k(\mathsf C,\mathsf S,\mathsf F)
=\mathsf C\otimes \mathsf S^{\otimes k}\otimes\mathsf F.
\end{equation}
The face maps use the right action, multiplication in $\mathsf S$, and the left action; degeneracy maps use the unit of $\mathsf S$.  The strict coequalizer of Sec.~\ref{sec:core} is the $0$th truncation, or under flatness/exactness hypotheses the ordinary realization, of this bar construction.

In that refinement, a DIS factorization proof object is not only a tuple of objects and maps.  It also contains coherence witnesses: associativity and unit homotopies for the module actions, homotopies witnessing balancedness of the collar map, homotopies showing commutativity of the measurement diagrams, and higher coherences among these witnesses.  The collar condition becomes
\begin{equation}
\Phi^{\Core}_{\coll}:
\mathsf C^{\leq N}\otimes^{\mathbb L}_{\mathsf S^{\leq N}}
\mathsf F^{\leq N}\to \A_{\DIS}(\Ucoll)
\end{equation}
being an $N$-equivalence in the homotopy category, or equivalently an equivalence after applying the chosen finite-accuracy truncation functor.

This perspective is useful for type-theoretic or computer-assisted formulations.  Each component of the proof object has a type: region system, stratified space, filtered object, algebra, module, balanced map, collar equivalence, and measurement map.  Each commutative diagram is a term witnessing an equality or homotopy of composites.  Missing region strata, nonclosed module sectors, or nonbalanced measurements become ill-typed or uninhabited pieces of the proof object rather than informal omissions.  The present paper keeps the strict model in the main text, but Eq.~\eqref{eq:derived-core} indicates the natural homotopy-coherent generalization.

\section{Lean-oriented theorem skeleton and checked finite model}
\label{app:lean}

The internal theorem in Sec.~\ref{sec:formal-theorem} is deliberately
small: it is a statement about a finite list of objects, maps,
isomorphisms, quotient objects, and commutative diagrams.  This appendix
records the proof-assistant interface suggested by that theorem.  The
purpose is not to formalize QCD analysis in Lean.  Rather, it is to show
that the formal spine of the argument separates cleanly from the analytic
collar-equivalence obligation.

There are two levels of formalization.  The first is the categorical
skeleton, which is the natural form of the theorem in an additive
symmetric monoidal category with the stated colimits.  The second is a
self-contained Lean model in which categories are replaced by ordinary
types, truncations are supplied explicitly, and the balanced core is
implemented as a quotient type.  This second model is not the final
mathematical setting of this paper, but it checks that the diagrammatic
logic, the balanced quotient, the scheme-invariance identity, and the
internal factorization theorem are executable in Lean.

\subsection{Categorical skeleton}

For mechanization in the strict categorical setting, one may ignore the
geometric origin of the data and retain only the following objects and
maps in an additive symmetric monoidal category $\V$ with the colimits
and exactness properties.

\begin{widetext}
Data:
\begin{align}
&\A_{\mathrm{all}},\A_{\coll},\mathsf C,\mathsf F\in\Fil(\V),
\qquad
\mathsf S\in\mathrm{Alg}(\Fil(\V)),\nonumber\\
&\mathsf C\in\mathrm{RMod}_{\mathsf S},
\qquad
\mathsf F\in\mathrm{LMod}_{\mathsf S},
\qquad
\Core=\mathsf C\otimes_{\mathsf S}\mathsf F,\nonumber\\
&\spc:\A_{\mathrm{all}}\to\A_{\coll},
\qquad
\Phi:\Core\to\A_{\coll},
\qquad
\bO:\mathbf 1\to\A_{\mathrm{all}},\nonumber\\
&\Meas_{\mathrm{all}}:\pi_{\leq N}\A_{\mathrm{all}}\to D,
\qquad
\Meas_{\coll}:\pi_{\leq N}\A_{\coll}\to D,
\qquad
\Meas_{\Core}:\pi_{\leq N}\Core\to D.
\end{align}
Axioms:
\begin{align}
&\Meas_{\mathrm{all}}
=
\Meas_{\coll}\circ\pi_{\leq N}(\spc),\nonumber\\
&\pi_{\leq N}\Phi
\text{ is an isomorphism},\nonumber\\
&\Meas_{\coll}\circ\pi_{\leq N}\Phi
=
\Meas_{\Core}.
\end{align}
Conclusion:
\begin{equation}
\Meas_{\mathrm{all}}\circ\pi_{\leq N}(\bO)
=
\Meas_{\Core}\circ(\pi_{\leq N}\Phi)^{-1}
\circ\pi_{\leq N}(\spc)
\circ\pi_{\leq N}(\bO).
\end{equation}
\end{widetext}

This is Theorem~\ref{thm:internal-dis}.  Its proof is rewriting by the
first axiom, insertion of the inverse of the truncated collar
isomorphism, and rewriting by the third axiom.

\subsection{Type-level Lean model}

A first executable model can be made without importing the category
theory library.  In this model, a filtered object is represented by a
carrier type, a family of truncated types, and a projection to each
truncation:
\begin{equation}
X
\quad\leadsto\quad
\bigl(
X_{\mathrm{car}},
\{\pi_{\leq N}X\}_{N\in\mathbb N},
\pi_N:X_{\mathrm{car}}\to \pi_{\leq N}X
\bigr).
\end{equation}
A filtered map $f:X\to Y$ is represented by an ordinary function on
carriers together with induced maps on every truncation,
\begin{equation}
f_N:\pi_{\leq N}X\to\pi_{\leq N}Y,
\end{equation}
and a commuting square
\begin{equation}
\pi_N^Y\circ f = f_N\circ\pi_N^X.
\end{equation}
Thus the Lean model does not construct cokernels.  It treats the
truncated object $\pi_{\leq N}X$ as explicit finite-accuracy data.

The interface algebra is represented by a monoid object at the level of
types:
\begin{equation}
(\mathsf S_N,1,\cdot),
\end{equation}
with a right action on coefficient presentations and a left action on
PDF/correlator presentations,
\begin{equation}
\mathsf C_N\curvearrowleft \mathsf S_N,
\qquad
\mathsf S_N\curvearrowright \mathsf F_N.
\end{equation}
The Lean structure includes the monoid laws and the left and right module
laws.  The balanced core is then implemented as the quotient type
\begin{equation}
\Core_N
=
\operatorname{Quot}\bigl((c\cdot s,f)\sim(c,s\cdot f)\bigr).
\end{equation}
The quotient map is denoted
\begin{equation}
q:\mathsf C_N\times\mathsf F_N\to\Core_N.
\end{equation}
Lean verifies the defining balancing identity
\begin{equation}
q(c\cdot s,f)=q(c,s\cdot f).
\end{equation}

A balanced coefficient--PDF pairing is represented by a function
\begin{equation}
\varphi:\mathsf C_N\times\mathsf F_N\to D
\end{equation}
satisfying
\begin{equation}
\varphi(c\cdot s,f)=\varphi(c,s\cdot f).
\end{equation}
Using the quotient eliminator, Lean constructs the descended measurement
\begin{equation}
\overline{\varphi}:\Core_N\to D
\end{equation}
and proves
\begin{equation}
\overline{\varphi}(q(c,f))=\varphi(c,f).
\end{equation}
This is the type-level form of the universal property of the balanced
core: every balanced evaluation factors through $\Core_N$.

\subsection{Scheme invariance in the Lean model}

The Lean model also checks finite scheme invariance.  A scheme change is
represented by an invertible element
\begin{equation}
z\in\mathsf S_N,
\qquad
z^{-1}z=1,
\qquad
zz^{-1}=1.
\end{equation}
The usual transformation
\begin{equation}
f'=z\cdot f,
\qquad
c'=c\cdot z^{-1}
\end{equation}
then satisfies
\begin{equation}
q(c',f')=q(c,f).
\end{equation}
The quotient calculation is the Lean analogue of the main text
scheme-covariance argument in Proposition~\ref{prop:scheme-covariance-core},
or equivalently of the core-level scheme-change calculation in
Eq.~\eqref{eq:scheme-demo-core}:
\begin{align}
q(c\cdot z^{-1},z\cdot f)
&=
q(c,z^{-1}\cdot(z\cdot f))
\nonumber\\
&=
q(c,(z^{-1}z)\cdot f)
\nonumber\\
&=
q(c,1\cdot f)
\nonumber\\
&=
q(c,f).
\end{align}
Consequently, every balanced measurement satisfies
\begin{equation}
\overline{\varphi}\bigl(q(c\cdot z^{-1},z\cdot f)\bigr)
=
\overline{\varphi}\bigl(q(c,f)\bigr),
\end{equation}
and equivalently
\begin{equation}
\varphi(c\cdot z^{-1},z\cdot f)=\varphi(c,f).
\end{equation}
This is the proof-assistant version of the statement that scheme-related
coefficient/PDF presentations define the same physical core element.

\subsection{Lean version of the internal DIS theorem}

The checked finite model then adds the DIS data
\begin{equation}
\pi_{\leq N}\A_{\mathrm{all}},
\qquad
\pi_{\leq N}\A_{\coll},
\qquad
\Core_N,
\end{equation}
a truncated specialization map
\begin{equation}
\mathrm{spc}_N \colon 
\pi_{\leq N} \A_{\mathrm{all}} 
\to 
\pi_{\leq N} \A_{\mathrm{coll}},
\end{equation}
and an isomorphism
\begin{equation}
\Phi_N:\Core_N\xrightarrow{\sim}\pi_{\leq N}\A_{\coll}.
\end{equation}
The global and collar measurements are functions
\begin{equation}
\Meas_{\mathrm{all}}:\pi_{\leq N}\A_{\mathrm{all}}\to D,
\qquad
\Meas_{\coll}:\pi_{\leq N}\A_{\coll}\to D,
\end{equation}
and the core measurement is the lifted balanced pairing
\begin{equation}
\Meas_{\Core}=\overline{\varphi}:\Core_N\to D.
\end{equation}
The Lean theorem states that if
\begin{align}
\mathsf{Meas}_{\mathrm{all}}(a)
&=
\mathsf{Meas}_{\mathrm{coll}}\!\left(\operatorname{spc}_{N}(a)\right),
\\
\mathsf{Meas}_{\mathrm{coll}}\!\left(\Phi_{N}(\xi)\right)
&=
\overline{\varphi}(\xi).
\end{align}
then for every observable point $x$,
\begin{equation}
\mathsf{Meas}_{\mathrm{all}}\!\bigl(\pi_{N}{\bO}(x)\bigr)
=
\overline{\varphi}\!\left(
{\Phi}_{N}^{-1}\!\left(
\operatorname{spc}_{N}\!\left(\pi_{N}{\bO}(x)\right)
\right)
\right).
\end{equation}
The proof is the same three-line diagram chase as
Theorem~\ref{thm:internal-dis}: rewrite through the collar, insert
$\Phi_N\Phi_N^{-1}=\mathrm{id}$, and rewrite using measurement
compatibility.

\subsection{Concrete toy instance}

The same Lean file can be populated by a finite toy proof object.  In
one such model,
\begin{equation}
\mathsf C_N=\mathsf S_N=\mathsf F_N=\{0,1\},
\end{equation}
with the interface monoid given by the two-element XOR operation.  The
left and right actions are the same XOR action, and the pairing
\begin{equation}
\varphi(c,f)=c\oplus f
\end{equation}
is balanced.  The nontrivial scheme element is its own inverse.  Taking
the collar equivalence to be the identity on the quotient core gives a
fully populated finite DIS datum.  Lean verifies:
\begin{align}
\mathsf{Meas}_{\mathrm{all}}\!\bigl(\pi_{N}{\bO}(x)\bigr)
&=
\mathsf{Meas}_{\Core}\!\left(
{\Phi}_{N}^{-1}\!\left(
\operatorname{spc}_{N}\!\left(\pi_{N}{\bO}(x)\right)
\right)
\right),
\\
q\!\left(c\cdot z^{-1},z\cdot f\right)
&=
q(c,f),
\\
\varphi\!\left(c\cdot z^{-1},z\cdot f\right)
&=
\varphi(c,f).
\end{align}
The toy instance is not a model of QCD.  Its role is to check that the
formal proof object is internally consistent and that the balanced core,
scheme invariance, and internal collar theorem can be executed in Lean.

\subsection{What remains beyond the finite Lean model}

The self-contained Lean model deliberately avoids several structures
needed for a full formalization:
\begin{enumerate}
\item cokernels defining $\pi_{\leq N}$ in a filtered abelian or
quasi-abelian category;
\item coequalizers and relative tensor products in a symmetric monoidal
category rather than quotient types of sets;
\item completed tensor products and strict quotients in the analytic
target $\V_{\mathrm{an}}$;
\item constructible factorization algebras on conically stratified
spaces;
\item the QCD proof that
$\Phi_{\coll}^{\Core}$ is an $N$-equivalence.
\end{enumerate}
Thus the Lean model verifies the formal core, not the analytic QCD
realization.  A full proof-assistant development would replace the
type-level quotient by categorical coequalizers, replace explicit
truncated types by cokernel-defined truncations, and then add the
geometric and analytic hypotheses required to construct the DIS regime
factorization algebra.  The finite model nevertheless confirms the
central structural claim: once the balanced core, the
truncated collar isomorphism, and measurement compatibility are supplied,
the factorized measured expression follows by formal computation.

\section{Notation dictionary}
\label{app:dictionary}

\begin{description}
\item[$K_{\DIS}$] Compactified stratified space of DIS asymptotic regimes.
\item[$\Rreg^{\pp}_{\DIS}$] Proof-level region poset, including hard, collinear, soft, jet, and overlap regions.
\item[$\Ucoll$] Collar neighborhood of the incoming-collinear boundary stratum.
\item[$\A_{\DIS}$] Constructible filtered regime factorization algebra.
\item[$\mathsf C$, $\mathsf F$] Hard coefficient module and collinear PDF/correlator module.
\item[$\mathsf S$] Interface algebra of finite collinear scheme transformations.
\item[$\Core$] Balanced carrier $\mathsf C\otimes_{\mathsf S}\mathsf F$.
\item[$\Phi^{\Core}_{\coll}$] Balanced collar comparison map.
\item[$\pi_{\leq N}$] Quotient modulo the power-suppressed piece $F^{N+1}$.
\item[$D$] Convolution algebra of measured momentum-fraction distributions.
\item[$\Van$] Admissible analytic target category for completed distributional QCD objects.
\end{description}

\end{document}